\documentclass[letterpaper,UKenglish,nolineno]{socg-lipics-v2019}
\usepackage{booktabs}

\usepackage{latexsym}

\usepackage[linesnumbered, vlined, ruled]{algorithm2e}
\usepackage{microtype}%if unwanted, comment out or use option "draft"

\usepackage[absolute]{textpos}

\usepackage{cite}

\def\calA{\mathcal{A}}

\theoremstyle{plain}
\newtheorem{observation}{Observation}
\newtheorem{mylemma}{Lemma}
\title{Algorithms for Subpath Convex Hull Queries and Ray-Shooting Among Segments\footnote{A preliminary version of this paper will appear in the Proceedings of the 36th International Symposium on Computational Geometry (SoCG 2020).}}

\titlerunning{Subpath Convex Hull Queries and Ray-Shooting}%optional, please use if title is longer than one line

\author{Haitao Wang}{Department of Computer Science, Utah State University, Logan, UT 84322, USA}{haitao.wang@usu.edu}{ https://orcid.org/0000-0001-8134-7409}{}%mandatory, please use full name; only 1 author per \author macro; first two parameters are mandatory, other parameters can be empty.

\authorrunning{H. Wang}%mandatory. First: Use abbreviated first/middle names. Second (only in severe cases): Use first author plus 'et al.'

\Copyright{Haitao Wang}%mandatory, please use full first names. LIPIcs license is "CC-BY";  http://creativecommons.org/licenses/by/3.0/

% mandatory: Please choose ACM 2012 classifications from https://www.acm.org/publications/class-2012 or https://dl.acm.org/ccs/ccs_flat.cfm . E.g., cite as "General and reference $\rightarrow$ General literature" or \ccsdesc[100]{General and reference~General literature}.

%mandatory

\hideLIPIcs  %uncomment to remove references to LIPIcs series (logo, DOI, ...), e.g. when preparing a pre-final version to be uploaded to arXiv or another public repository
\begin{document}
%\begin{textblock}{5}(2.2,1.8)
%\noindent\LARGE {\bf APPENDIX}
%\end{textblock}

\maketitle

\begin{abstract}
In this paper, we first consider the subpath convex hull query problem:
Given a simple path $\pi$ of $n$ vertices, preprocess it so that the
convex hull of any query subpath of $\pi$ can be quickly obtained. Previously, Guibas, Hershberger, and
Snoeyink [SODA 90'] proposed a data structure of $O(n)$ space and $O(\log n\log\log n)$ query time;
reducing the query time to $O(\log n)$ increases
the space to $O(n\log\log n)$.
%We present an improved result for the first time in the past three decades.
We present an improved result that uses $O(n)$ space while achieving $O(\log n)$ query
time. Like the previous work, our query algorithm returns a compact interval tree representing the convex hull
so that standard binary-search-based queries on the hull can be
performed in $O(\log n)$ time each. The preprocessing time of our data structure is $O(n)$, after the
vertices of $\pi$ are sorted by $x$-coordinate. As the subpath convex hull query problem has many applications, our new result leads to improvements for several other problems.

In particular, with the help of the above result, along with other
techniques, we present new algorithms for the ray-shooting problem among
segments. Given a set of $n$ (possibly intersecting) line segments in
the plane, preprocess it so that the first segment
hit by a query ray can be quickly found. We give a data structure
of $O(n\log n)$ space that can answer each query in $(\sqrt{n}\log n)$ time.
%This is the first data structure of $(n\log n)$ space and
%$\widetilde{O}(\sqrt{n})$\footnote{The notation $\widetilde{O}$ hides
%a polylogarithmic factor} query time.
If the segments are nonintersecting or if the segments are lines, then the space can be reduced to $O(n)$.
As a by-product,
%given a set of $n$ lines in the plane, we build a data structure of
%$O(n)$ space that can compute the first line hit by a query ray in
%$O(\sqrt{n}\log n)$ time;
given a set of $n$ (possibly intersecting) segments in the plane,
we build a data structure of $O(n)$ space that can determine whether a query line intersects a segment in
$O(\sqrt{n}\log  n)$ time. The preprocessing time is $O(n^{1.5})$ for all four problems,
which can be reduced to $O(n\log n)$ time by a randomized algorithm so that the
query time is bounded by $O(\sqrt{n}\log n)$ with high probability.
All these are classical problems that have been studied extensively. Previously data
structures of $\widetilde{O}(\sqrt{n})$ query time\footnote{The notation $\widetilde{O}$ suppresses a polylogarithmic factor.} were known in early 1990s; nearly no progress has been made for more than two decades. For all these problems, our new results provide improvements by
reducing the space of the data structures by at least a logarithmic
factor while the preprocessing and query times are the same as before or even better.
%In particular, our result for the problem is the first-known data structure of $O(n\log n)$ space with $\widetilde{O}(\sqrt{n})$ query time; our results for the other problems are first-known data structures of $O(n)$ space with $\widetilde{O}(\sqrt{n})$ query time.
\end{abstract}

\section{Introduction}
\label{sec:intro}
In this paper, we first consider the subpath convex hull query
problem. Let $\pi$ be a simple path of $n$ vertices in the plane. A
{\em subpath hull query} specifies two vertices of $\pi$ and asks for
the convex hull of the subpath between the two vertices. The goal is to preprocess $\pi$ so that the subpath hull queries
can be answered quickly. Ideally, the query
should return a representation of the convex hull
so that standard queries on the hull can be performed in logarithmic
time.

The problem has been studied by Guibas, Hershberger, and
Snoeyink\cite{ref:GuibasCo91}, who proposed a method of using
compact interval trees. After $O(n\log n)$ time
preprocessing, Guibas et al.~\cite{ref:GuibasCo91} built a data
structure of $O(n)$ space that can answer each query in
$O(\log n\log\log n)$ time. Their query algorithm returns a compact interval
tree that represents the convex hull so that all binary-search-based
queries on the hull can be performed in $O(\log n)$ time each. The
queries on the hull include (but are not limited to) the following: find
the most extreme vertex of the convex hull along a query direction;
find the intersection between a query line and the convex hull; find
the common tangents from a query point to the convex hull; determine
whether a query point is inside the convex hull, etc.
Guibas et al.~\cite{ref:GuibasCo91} reduced the subpath hull
query time to $O(\log n)$ but the space becomes $O(n\log\log n)$. A
trade-off was also made with $O(\log n\log^*n)$ query time and
$O(n\log^*n)$ space~\cite{ref:GuibasCo91}.

As compact interval trees are quite amenable, the results of Guibas et
al.~\cite{ref:GuibasCo91} have found many applications,
e.g.,~\cite{ref:BeckerEn96,ref:ChengAl92,ref:ChenEf15,ref:ChenAp13,ref:ChristIm10,ref:DaiOp13,ref:WangOn19}.
Clearly, there is still some room for further improvement on the
results of Guibas et al.~\cite{ref:GuibasCo91}; the ultimate goal
might be an $O(n)$ space data structure with $O(\log n)$ query time.
In this paper, we achieve this goal. The preprocessing time of our data
structure is $O(n)$, after the vertices of $\pi$ are sorted by $x$-coordinate.
Like the results of Guibas et al.~\cite{ref:GuibasCo91}, our query
algorithm also returns a compact interval tree that can support
logarithmic time queries for all binary-search-based queries on the
convex hull of the query subpath; the edges of the convex hull can be
retrieved in time linear in the number of vertices of the convex hull. Note that like those in~\cite{ref:GuibasCo91} our results are for the random access machine (RAM) model.

With our new result, previous applications that use the results of
Guibas et al.~\cite{ref:GuibasCo91} can now be improved accordingly. We will demonstrate some of
them, including the problem of enclosing polygons by two
minimum area rectangles~\cite{ref:BeckerAn91,ref:BeckerEn96}, computing a guarding set for simple polygons in wireless location~\cite{ref:ChristIm10}, computing optimal time-convex
hulls~\cite{ref:DaiOp13}, $L_1$ top-$k$ weighted sum aggregate nearest and
farthest neighbor searching~\cite{ref:WangOn19}, etc.
%the 1D $k$-center problem~\cite{ref:ChenEf15}, approximating points by piecewise linear functions~\cite{ref:ChenAp13}, etc.
For all these problems, we reduce
the space of their algorithms by a $\log\log n$ factor while
the time complexities are the same as before or even better.
%unaltered asymptotically.

We should point out that Wagener~\cite{ref:WagenerOp92} proposed a parallel algorithm for computing a data structure, called {\em bridge tree}, for representing the convex hull of a simple path $\pi$. If using one processor, for any query subpath of $\pi$, Wagener~\cite{ref:WagenerOp92} showed that the bridge tree can be used to answer decomposable queries\footnote{A convex hull query is decomposable if the answer to the query on a point set S can be obtained in constant time from the answers to the queries on $S_1$ and $S_2$, where $S_1$ and $S_2$ form a disjoint partition of $S$. For example, the following queries are decomposable: find the most extreme vertex of the convex hull along a query direction; find the two common tangents to the convex hull from a query point outside the hull, while the following queries are not decomposable: find the intersection of the convex hull with a query line; find the common tangents for two disjoint convex hulls.} on the convex hull of the query subpath in logarithmic time each. Wagener~\cite{ref:WagenerOp92} claimed that some non-decomposable queries can also be handled; however no details were provided. In contrast, our approach returns a compact interval tree that is more amenable (indeed, the bridge trees~\cite{ref:WagenerOp92} were mainly designed for parallel processing) and can support both decomposable and non-decomposable queries. In addition, if one wants to output the convex hull of the query subpath, our approach can do so in time linear in the number of the vertices of the convex hull while the method of Wagener~\cite{ref:WagenerOp92} needs $O(n)$ time.

%We should point that Wagener~\cite{ref:Wagener} proposed a parallel
%algorithm for computing a data structure, called {\em bridge tree},
%for representing the convex hull of a simple path $\pi$. If using one
%processor, his algorithm computes in $O(n)$ time a bridge tree so that
%any decomposable queries\footnote{A query is decomposable if } on the
%convex hull of $\pi$ can be answered in logarithmic time. Wagener~\cite{ref:Wagener} also
%claimed without proving any details that his data structure can also support logarithmic time
%subpath hull queries. Note that the compact interval trees obtained by
%our subpath hull queries can also handle non-decomposable queries.

\subsection{Ray-Shooting}
\label{sec:introray}

With the help of our subpath hull query data structure and many other new techniques, we present improved results for several classical ray-shooting problems.  These problems have been studied extensively. Previously, data structures of $\widetilde{O}(\sqrt{n})$ query time and near-linear space were known in early 1990s; nearly no progress has been made for over two decades. Our new results reduce the space by at least a logarithmic factor while still achieving the same or even better preprocessing and query times.

In the following, we use a triple $(T(n),S(n),Q(n))$ to represent the complexity of a data structure, where $T(n)$ is the preprocessing time, $S(n)$ is the space, and $Q(n)$ is the query time. We will confine the discussion of the previous work to data structures of linear or near-linear space.
Refer to Table~\ref{tab:10} for a summary. Throughout the paper, we use $\delta$ to refer to an arbitrarily small positive constant.

%\begin{definition}[Ray-shooting among lines]
%Given a set of $n$ lines in the plane, the problem is to build a data structure so that the first line hit by a query ray can be found efficiently.
%\end{definition}

\subparagraph{Ray-shooting among lines.} Given a set of $n$ lines in the plane, the problem is to build a data structure so that the first line hit by a query ray can be quickly found.

%For this problem,
Bar-Yehuda and Fogel~\cite{ref:Bar-YehudaVa94} gave a data structure
of complexity $O(n^{1.5},n\log^2 n,\sqrt{n}\log n)$. Cheng and
Janardan~\cite{ref:ChengAl92} gave a data structure of complexity
$O(n^{1.5}\log^2 n,n\log n,\sqrt{n}\log n)$. Agarwal and
Sharir~\cite{ref:AgarwalAp93} developed a data structure of complexity
$O(n\log n, n\log n, n^{1/2+\delta})$.
%for any small constant $\delta>0$.

By using our subpath hull query data structure and a result from Chazelle and Guibas~\cite{ref:ChazelleFr862}, we present a new data structure of complexity $O(n^{1.5},n,\sqrt{n}\log n)$. This is the first time that this problem is solved in $\widetilde{O}(\sqrt{n})$ time while using only $O(n)$ space.
%The preprocessing time can be reduced to $O(n\log n)$ using an algorithm from Chan~\cite{ref:ChanOp12}, but the query time becomes $O(\sqrt{n}\log n)$ expected time.

In addition, we also consider a more general {\em first-$k$-hits} query, i.e., given a query ray and an integer $k$, report the first $k$ lines hit by the ray. This problem was studied by Bar-Yehuda and Fogel~\cite{ref:Bar-YehudaVa94}, who gave a data structure of complexity $O(n^{1.5},n\log^2 n,\sqrt{n}\log n+k\log^2 n)$. Our new result is a data structure of complexity $O(n^{1.5},n,\sqrt{n}\log n+k\log n)$.

%begin{Definition}[Intersection Detection Among Segments For Query Lines]
%Given A Set Of $N$ Segments In The Plane, The Problem Is To Build A Data Structure To Determine Whether A Query Line Intersects At Least One Segment.
%\End{Definition}

\subparagraph{Intersection detection.} Given a set of $n$ line segments in the plane, the problem is to build a data structure to determine whether a query line intersects at least one segment.

Cheng and Janardan~\cite{ref:ChengAl92} gave a data structure of complexity $O(n^{1.5}\log^2 n,n\log n,\sqrt{n}\log n)$. By adapting the interval partition trees of Overmars et al.~\cite{ref:OvermarsSt90} (which relies on the conjugation trees of Edelsbrunner and Welzl~\cite{ref:EdelsbrunnerHa86}) to the partition trees of Matou\v{s}ek~\cite{ref:MatousekEf92,ref:MatousekRa93}, we obtain a data structure of complexity $O(n^{1.5},n,\sqrt{n}\log  n)$. To this end, we have to use Matou\v{s}ek's techniques in both \cite{ref:MatousekEf92} and \cite{ref:MatousekRa93}, and modify them in a not-so-trivial mannar.
%The preprocessing time can be reduced to $O(n\log n)$ and the query time becomes $O(\sqrt{n}\log n)$ expected time by using Chan's technique~\cite{ref:ChanOp12}.

%\begin{definition}[Ray-shooting among segments]
%Given a set of $n$ (possibly intersecting) line segments in the plane, the problem is to build a data structure so that the first segment hit by a query ray can be found efficiently.
%\end{definition}

\subparagraph{Ray-shooting among segments.}
Given a set of $n$ (possibly intersecting) line segments in the plane, the problem is to build a data structure to find the first segment hit by a query ray.

Overmars et al.~\cite{ref:OvermarsSt90} gave a data structure of complexity $O(n\alpha(n)\log^3 n, n\log^2 n, n^{0.695}\log n)$, where $\alpha(n)$ is the inverse Ackermann's function. Guibas et al.~\cite{ref:GuibasIn88} presented a data structure of complexity $O(n\alpha(n)\log^3 n, n\alpha(n), n^{2/3+\delta})$. Agarwal~\cite{ref:AgarwalRa922} gave a data structure of complexity $O(n^{1.5}\log^{4.33}n,n\alpha(n)\log^4 n, \sqrt{n\alpha(n)}\log^2 n)$.
%and another data structure of complexity $O(n^2\alpha^2(n)\log n, n^2\alpha^2(n), \log n)$. Tradeoffs between preprocessing and query time are also provided in~\cite{ref:AgarwalRa922}.
Bar-Yehuda and Fogel~\cite{ref:Bar-YehudaVa94} gave a data structure of complexity $O((n\alpha(n))^{1.5},n\alpha(n)\log^2 n,\sqrt{n\alpha(n)}\log n)$.
Cheng and Janardan~\cite{ref:ChengAl92} developed a data structure of
complexity $O(n^{1.5}\log^2 n,n\log^2 n,\sqrt{n}\log n)$. Agarwal and
Sharir~\cite{ref:AgarwalAp93} proposed a data structure of complexity
$O(n\log^2 n,n\log^2 n,n^{0.5+\delta})$.
Chan's randomized techniques~\cite{ref:ChanOp12} yielded a data structure of complexity $O(n\log^3 n, n\log^2 n, \sqrt{n}\log^2 n)$, where the query time is expected.

Cheng and Janardan's algorithm~\cite{ref:ChengAl92} relies on their
results for the ray-shooting problem among lines and the intersection detection problem.
Following their algorithmic
scheme and using our above new results for these two problems, we obtain a
data structure for the ray-shooting problem among segments with
complexity $O(n^{1.5},n\log n,\sqrt{n}\log n)$. This is the first
data structure of $\widetilde{O}(\sqrt{n})$ query time
that uses only $O(n\log n)$ space.
%Again, the preprocessing time can be reduced to $O(n\log n)$ and the query time becomes $O(\sqrt{n}\log n)$ expected time.

If the segments are nonintersecting, then better
results exist. Overmars et al.~\cite{ref:OvermarsSt90} gave a data
structure of complexity $O(n\log n, n, n^{0.695}\log n)$.
Agarwal~\cite{ref:AgarwalRa922} presented a data structure
of complexity $O(n^{1.5}\log^{4.33}n,n\alpha(n)\log^3 n,
\sqrt{n}\log^2 n)$.
%and another data structure of complexity $O(n^2\log n, n^2, \log n)$. Tradeoffs between
%preprocessing and query time are also provided in~\cite{ref:AgarwalRa922}.
Bar-Yehuda and Fogel~\cite{ref:Bar-YehudaVa94} proposed a data structure
of complexity $O(n^{1.5},n\log n,\sqrt{n}\log n)$.
%By adapting the techniques of Overmars et al.~\cite{ref:OvermarsSt90} to the partition
%trees of Matou\v{s}ek~\cite{ref:MatousekEf92,ref:MatousekRa93},
Our new data structure has complexity $O(n^{1.5},n,\sqrt{n}\log n)$.
%The preprocessing time can be reduced to $O(n\log n)$ time and the query time becomes $O(\sqrt{n}\log n)$ expected time.
This is the first data structure of $\widetilde{O}(\sqrt{n})$ query time
that uses only $O(n)$ space.
Note that if the segments form the boundary of a simple polygon, then there exist data structures of complexity $O(n,n,\log n)$~\cite{ref:ChazelleRa94,ref:ChazelleVi89,ref:HershbergerA95}.

\begin{table}[t]
{
\begin{center}
\begin{tabular}{l l l l l}
\toprule[0.02in]
 & Preprocessing time & Space & Query time & Source \\
%\multicolumn{3}{ |c| }{Team sheet} \\
\midrule
\multirow{4}{2.5cm}{Ray-shooting among lines}
 & $n^{1.5}$ & $n\log^2 n$ & $\sqrt{n}\log n$ & BF\cite{ref:Bar-YehudaVa94}\\
 & $n^{1.5}\log^2 n$ & $n\log n$ & $\sqrt{n}\log n$ & CJ\cite{ref:ChengAl92}\\
 & $n\log n$ & $n\log n$ & $n^{0.5+\delta}$ &  AS\cite{ref:AgarwalAp93}\\
 & $n^{1.5}$ & $n$ & $\sqrt{n}\log n$ & this paper\\
 & $n\log n$  & $n$ & $\sqrt{n}\log n$ * & this paper\\
 \hline
\multirow{3}{2.5cm}{Intersection detection}
 & $n^{1.5}\log^2 n$ & $n\log n$ & $\sqrt{n}\log n$ & CJ\cite{ref:ChengAl92}\\
 & $n^{1.5}$ & $n$ & $\sqrt{n}\log  n$ & this paper\\
 & $n\log n$  & $n$ & $\sqrt{n}\log n$ * & this paper\\
 \hline
\multirow{9}{2.5cm}{Ray-shooting among intersecting segments}
& $n\alpha(n)\log^3 n$ & $n\log^2 n$  &  $n^{0.695}\log n$ & OSS\cite{ref:OvermarsSt90}\\
& $n\alpha(n)\log^3 n$ & $n\alpha(n)$ & $n^{2/3+\delta}$ & GOS\cite{ref:GuibasIn88}\\
& $n^{1.5}\log^{4.33}n$& $n\alpha(n)\log^4 n$ & $\sqrt{n\alpha(n)}\log^2 n$ & A~\cite{ref:AgarwalRa922}\\
& $(n\alpha(n))^{1.5}$ & $n\alpha(n)\log^2 n$ & $\sqrt{n\alpha(n)}\log n$ & BF~\cite{ref:Bar-YehudaVa94}\\
 & $n^{1.5}\log^2 n$ & $n\log^2 n$  & $\sqrt{n}\log n$ & CJ\cite{ref:ChengAl92}\\
 & $n\log^2 n$ & $n\log^2 n$ & $n^{0.5+\delta}$ & AS~\cite{ref:AgarwalAp93}\\
 & $n\log^3 n$  & $n\log^2 n$ & $\sqrt{n}\log^2 n$ * & C~\cite{ref:ChanOp12} \\
 & $n^{1.5}$ & $n\log n$ & $\sqrt{n}\log  n$ & this paper\\
 & $n\log^2 n$  & $n\log n$ & $\sqrt{n}\log n$ * & this paper\\
 \hline
\multirow{5}{2.5cm}{Ray-shooting among nonintersecting segments}
& $n\log n$ & $n$ & $n^{0.695}\log n$ &  OSS~\cite{ref:OvermarsSt90}\\
& $n^{1.5}\log^{4.33}n$ & $n\alpha(n)\log^3 n$ & $\sqrt{n}\log^2 n$ & A~\cite{ref:AgarwalRa922}\\
& $n^{1.5}$ & $n\log n$ & $\sqrt{n}\log n$  & BF~\cite{ref:Bar-YehudaVa94}\\
& $n^{1.5}$ & $n$ & $\sqrt{n}\log  n$ & this paper\\
& $n\log n$ & $n$ & $\sqrt{n}\log n$ * & this paper\\
\bottomrule[0.02in]
\end{tabular}
\vspace{0.2in}
\caption{\footnotesize Summary of the results. The big-$O$ notation is omitted. $\delta$ can be any small positive constant. The results marked with * hold with high probability (except that the result of Chan~\cite{ref:ChanOp12} is expected).} \label{tab:10}
\end{center}
}
\vspace*{-0.25in}
\end{table}

%\medskip
\subparagraph{Randomized results.}
Using Chan's randomized techniques~\cite{ref:ChanOp12}, the preprocessing time of all our above results can be reduced to $O(n\log n)$ (except $O(n\log^2 n)$ time for the ray-shooting problem among intersecting segments), while the same query time complexities hold with high probability (i.e., probability at least $1 - 1/n^c$ for  any large constant $c$).

\subparagraph{Outline.} The rest of the paper is organized as follows. In Section~\ref{sec:pre} we review some previous work of the subpath hull query problem; Section~\ref{sec:result} presents our new data structure for the problem. Section~\ref{sec:ray} is concerned with the ray-shooting problem.
Other applications of the our subpath hull query result are discussed in Section~\ref{sec:app}.

\section{Preliminaries}
\label{sec:pre}

Let $p_1,\ldots,p_n$ be the vertices of a simple path $\pi$
ordered along $\pi$.
%We often refer to these vertices and their indices interchangeably, e.g., vertex $i$ refers to $p_i$.
For any two indices $i$ and $j$ with $1\leq i\leq j\leq n$, we use
$\pi(i,j)$ to refer to the subpath of $\pi$ from $p_i$ to $p_j$.
Given a pair $(i,j)$ of indices
with $1\leq i\leq j\leq n$, the {\em subpath hull query} asks for the convex
hull of $\pi(i,j)$.

The convex hull of a simple path
can be found in linear time, e.g.,~\cite{ref:GrahamFi83,ref:MelkmanOn87}.
Note that the convex hull of a simple path is the same as the convex hull of its vertices. For this reason, in our discussion a subpath $\pi'$ of $\pi$ actually refers to its vertex set.
%Because the convex hull of $\pi(i,j)$ is the same as the convex hull
%of the vertices of $\pi(i,j)$, although $\pi(i,j)$ consists of edges,
%we actually use $\pi(i,j)$ to refer to the vertex set $\{p_i,p_{i+1},\ldots,p_j\}$.
For each subpath $\pi'$ of $\pi$, we use $|\pi'|$ to denote the number
of vertices of $\pi'$; we consider the endpoint of $\pi'$ that is
closer to $p_1$ in $\pi$ as the {\em first vertex} of $\pi'$ while the
other endpoint is the {\em last vertex} of $\pi'$. So $p_i$ is the first
vertex and $p_j$ is the last vertex of $\pi(i,j)$.

For any set $P$ of points in the plane, let $H(P)$ denote the convex
hull of $P$. Denote by $H_U(P)$ and $H_L(P)$ the upper and lower
hulls, respectively.

\subparagraph{Interval trees.}
Let $S$ be a set of $n$ points in the plane. The {\em interval tree} $T(S)$
is a complete binary tree whose leaves
from left to right correspond to the points of $S$ sorted from left
to right. Each internal node corresponds to the interval between the
rightmost leaf in its left subtree and the leftmost leaf in its right
subtree. We say that a segment joining two points of $S$ {\em spans} an
internal node $v$ if $v$ is between the two endpoints of the
segment in the symmetric order of the nodes of $T(S)$ (or equivalently, the
projection of the interval of $v$ on the $x$-axis is contained in the
projection of the segment on the $x$-axis).

We store each edge $e$ of the upper hull $H_U(S)$ at the highest node of $T(S)$
that $e$ spans (e.g., see Fig.~\ref{fig:intervaltree}). By also storing the edges of the lower hull $H_L(S)$ in
$T(S)$ in the same way, we can answer all standard binary-search-based
queries on the convex hull $H(S)$ in $O(\log n)$ time,
by following a path from the root of $T(S)$ to a leaf~\cite{ref:GuibasCo91}. The main
idea is that the edge of $H_U(S)$ (resp., $H_L(S)$) spanning a node
$v$ of $T(S)$ is stored either at $v$ or at one of $v$'s ancestors and
only at most two ancestors
closest to $v$ (one to the left and the other to the right of $v$)
need to be remembered during the
search (see Lemma 4.1 of~\cite{ref:GuibasCo91} for details).

\begin{figure}[t]
\begin{minipage}[t]{\textwidth}
\begin{center}
\includegraphics[height=1.5in]{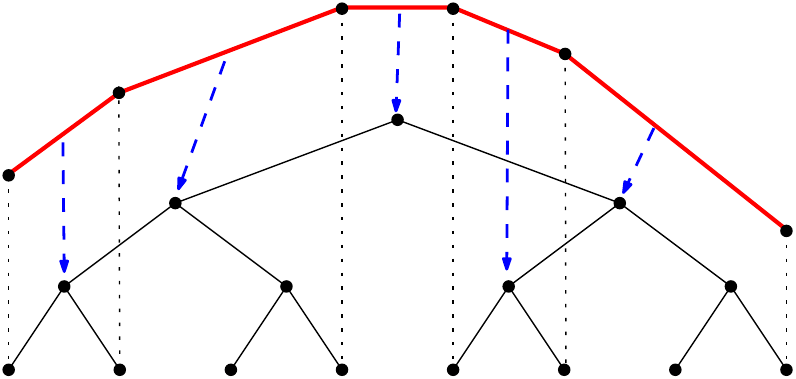}
\caption{\footnotesize Illustrating an interval tree that stores upper hull edges: the (blue) dashed lines with arrows indicate where edges are stored.}
\label{fig:intervaltree}
\end{center}
\end{minipage}
\vspace{-0.15in}
\end{figure}

\subparagraph{Compact interval trees.}
As the size of $T(S)$ is $\Theta(n)$ while $|H(S)|$ may be much smaller
than $n$, where $|H(S)|$ is the number of edges of $H(S)$,
%it may not always be a good idea to
using $T(S)$ to store $H(S)$ may not be space-efficient.
Guibas et al.~\cite{ref:GuibasCo91} proposed to use a {\em compact interval tree} $T_U(S)$ of $O(|H_U(S)|)$ size to
store $H_U(S)$, as follows. In $T(S)$, a node $v$ is {\em
empty} if it does not store an edge of $H_U(S)$; otherwise it is {\em
full}. It was shown in~\cite{ref:GuibasCo91} that if two
nodes of $T(S)$ are full, then their lowest common ancestor is also
full. We remove empty nodes from $T(S)$ by relinking the tree to make
each full node the child of its nearest full ancestor. Let $T_U(S)$ be
the new tree and we still use $T(S)$ to refer to the original interval tree
without storing any hull edges.
Each node of $T_U(S)$ stores exactly one edge of $H_U(S)$, and thus
$T_U(S)$ has $|H_U(S)|$ nodes. After $O(n)$ time preprocessing on
$T(S)$ (specifically, build a lowest common ancestor query
data structure~\cite{ref:BenderTh00,ref:HarelFa84}, with constant query time), $T_U(S)$ can be computed from $H_U(S)$ in
$O(|H_U(S)|)$ time (see Lemma 4.4 in~\cite{ref:GuibasCo91}).
Similarly, we use a compact interval tree $T_L(S)$ of $|H_L(S)|$ nodes
to store $H_L(S)$. Then, using the three trees $T_U(S)$, $T_L(S)$, and
$T(S)$, all standard binary-search-based queries on $H(S)$ can be answered
in $O(\log n)$ time. The main idea is that the algorithm walks down
through the compact interval trees while keeping track of the
corresponding position in $T(S)$ (see Lemma 4.3~\cite{ref:GuibasCo91}
for details). We call $T(S)$ a {\em reference tree}. In addition, using $T_U(S)$ and $T_L(S)$, $H(S)$ can be output in $O(|H(S)|)$ time.

As discussed above, to represent $H(S)$, we need two compact interval
trees, one for $H_U(S)$ and the other for $H_L(S)$. To make our
discussion more concise, we will simply say ``the compact interval tree''
for $S$ and use $T^+(S)$ to refer to it, which actually includes two
trees.

%Suppose we have the compact interval trees $T^+(S_1)$ and $T^+(S_2)$ for
%two disjoint subsets $S_1$ and $S_2$ of $S$, as well as
%$T(S)$. Also suppose that $H(S_1)$ and $H(S_2)$ have only one common
%(upper) tangent, which is known to us. Then, $T^+(S_1\cup S_2)$ can be
%obtained in $O(\min(k,\log n))$ time (Lemma 4.5
%of~\cite{ref:GuibasCo91}), where $k$ is the total number of vertices
%in $U(S_1)$ and $U(S_2)$.

\subparagraph{Compact interval trees for $\boldsymbol{\pi}$.}
Consider two {\em consecutive} subpaths $\pi_1$ and $\pi_2$ of $\pi$.
Suppose their compact interval trees $T^+(\pi_1)$ and
$T^+(\pi_2)$ as well as the interval tree $T(\pi)$ of $\pi$ are available.
It is known that the convex hulls of two consecutive subpaths of a simple path
have at most two common tangents~\cite{ref:ChazelleFr862}.
Hence, $H(\pi_1)$ and $H(\pi_2)$ have at most two common tangents.
By using the path-copying method of persistent data
structures~\cite{ref:DriscollMa89}, Guibas et al.~\cite{ref:GuibasCo91} obtained the following result.

\begin{mylemma}\label{lem:merge}{\em (Guibas et al.~\cite{ref:GuibasCo91})}
Without altering $T^+(\pi_1)$ and $T^+(\pi_2)$, the compact interval
tree $T^+(\pi_1\cup \pi_2)$ can be produced (the root of the
tree will be returned) in $O(\log n)$ time and
$O(\log n)$ additional space.
\end{mylemma}

\begin{mylemma}\label{lem:subpathhull}{\em (Guibas et al.~\cite{ref:GuibasCo91})}
Given the interval tree $T(\pi)$, with $O(n)$ time preprocessing, we can compute $T^+(\pi')$ for any subpath $\pi'$ of $\pi$ in $O(|\pi'|)$ time.
\end{mylemma}
\begin{proof}
We preprocess $T(\pi)$ in the same way as preprocessing $T(S)$ discussed before (i.e., build a lowest common ancestor query data structure~\cite{ref:BenderTh00,ref:HarelFa84}, with constant query time).
For any subpath $\pi'$ of $\pi$, we first compute its convex hull $H(\pi')$ in $O(|\pi'|)$ time~\cite{ref:GrahamFi83,ref:MelkmanOn87}. Then, as discussed before, $T^+(\pi')$ can be constructed in $O(|H(\pi')|)$ time (Lemma 4.4 in~\cite{ref:GuibasCo91}).
\end{proof}

\section{Subpath Convex Hull Queries}
\label{sec:result}

In this section, we present our new data structure for subpath
hull queries. We first compute a sorted list of all vertices of $\pi$ by $x$-coordinate.
As will be seen later, the rest of the preprocessing of our data structure takes $O(n)$ time in total.

\subsection{A decomposition tree}

After having the interval tree $T(\pi)$, we construct a {\em decomposition tree} $\Psi(\pi)$, which is a segment tree on
the vertices of $\pi$ following their order along $\pi$.
Specifically, $\Psi(\pi)$ is a complete binary tree with $n$ leaves
corresponding to the vertices of $\pi$ in order along $\pi$.
%In the following, we refer to the leaves of $\Psi(\pi)$ and the indices of the vertices of $\pi$ interchangeably.
Each internal node $v$ of $\Psi(\pi)$ corresponds to the subpath $\pi(a_v,b_v)$, where $a_v$
(resp., $b_v$) is defined to be the index of the vertex of $\pi$ corresponding to the leftmost (resp.,
rightmost) leaf of the subtree of $\Psi(\pi)$ rooted at $v$; we call $\pi(a_v,b_v)$
a {\em canonical subpath} of $\pi$ and use $\pi(v)$ to
denote it.
%Hence, each node of $\Psi(\pi)$ corresponds to a canonical subpath.
%We sometime use canonical subpaths and nodes of $\Psi(\pi)$ interchangeably.

Next, we remove some nodes in the lower part of $\Psi(\pi)$, as follows. For each node $v$ whose canonical path has at most $\log^2 n$ vertices and whose parent canonical subpath has more than $\log^2 n$ vertices, we remove both the left and the right subtrees of $v$ from $\Psi(\pi)$ but explicitly store $\pi(v)$ at $v$, after which $v$ becomes a leaf of the new tree. From now on we use $\Psi(\pi)$ to refer to the new tree. It is not difficult to see that $\Psi(\pi)$ now has $O(n/\log^2 n)$ nodes.

We then compute compact interval trees $T^+(\pi(v))$ for all nodes $v$ of $\Psi(\pi)$ in a bottom-up manner. Specifically, if $v$ is a leaf, then $\pi(v)$ has at most $\log^2n$ vertices, and we compute $T^+(\pi(v))$ from scratch, which takes $O(\log^2 n)$ time by Lemma~\ref{lem:subpathhull}. If $v$ is not a leaf, then $T^+(\pi(v))$ can be obtained by merging the two compact interval trees of its children, which takes $O(\log n)$ time by Lemma~\ref{lem:merge}. In this way, computing compact interval trees for all nodes of $\Psi(\pi)$ takes $O(n)$ time in total, for $\Psi(\pi)$ has $O(n/\log^2 n)$ nodes.

\subsection{A preliminary query algorithm}

Consider a subpath hull query $(i,j)$. We first present an $O(\log^2
n)$ time query algorithm using $\Psi(\pi)$ and then reduce
the time to $O(\log n)$. Depending on whether the two vertices $p_i$ and $p_j$ are in
the same canonical subpath of a leaf of $\Psi(\pi)$, there are two
cases.

\begin{description}
\item[Case 1.]
If yes, let $v$ be the leaf. Then, $\pi(i,j)$ is a subpath of $\pi(v)$ and thus has at most $\log^2 n$ vertices. We compute $T^+(\pi(i,j))$ from scratch in $O(\log^2 n)$ time by Lemma~\ref{lem:subpathhull}.

\item[Case 2.]
Otherwise, let $v$ be the leaf of $\Psi(\pi)$ whose canonical subpath contains $p_i$ and $u$ the leaf whose canonical subpath contains $p_j$. Let $w$ be the lowest common ancestor of $u$ and $v$.
As in~\cite{ref:GuibasCo91}, we partition $\pi(i,j)$ into two subpaths
$\pi(i,k)$ and $\pi(k+1,j)$, where $k=b_{w'}$ with $w'$ being the left
child of $w$ (recall the definition of $b_{w'}$ given before).
We will compute the compact interval trees for the
two subpaths separately, and then merge them to obtain $T^+(\pi(i,j))$
in additional $O(\log n)$ time by Lemma~\ref{lem:merge}. We only
discuss how to compute $T^+(\pi(i,k))$, for the other tree can be
computed likewise.

We further partition $\pi(i,k)$ into two subpaths $\pi(i,b_v)$ and
$\pi(b_v+1,k)$. We will compute the compact interval trees for them separately and then merge the two trees to obtain $T^+(\pi(i,k))$.

For computing $T^+(\pi(i,b_v))$, as $\pi(i,b_v)$ is a subpath of
$\pi(v)$, it has at most $\log^2 n$ vertices. Hence, we can compute
$T^+(\pi(i,b_v))$ from scratch in $O(\log^2 n)$ time.

For computing $T^+(\pi(b_v+1,k))$, observe that $\pi(b_v+1,k)$ is the
concatenation of the canonical subpaths of $O(\log n)$ nodes of
$\Psi(\pi)$; precisely, these nodes are the right children of their
parents that are in the path of $\Psi(\pi)$ from $v$'s parent to
$w'$ and these nodes themselves are not on the path. Since the compact interval trees of these nodes are already available
due to the preprocessing, we can produce $T^+(\pi(b_v+1,k))$ in
$O(\log^2 n)$ time by merging these trees.
%with $O(\log n)$ time per merge by Lemma~\ref{lem:merge}.
\end{description}

In summary, we can compute $T^+(\pi(i,j))$ in $O(\log^2 n)$ time in
either case.

\subsection{Reducing the query time to $O(\log n)$}

In what follows, we reduce the query time to
$O(\log n)$, with additional preprocessing (but still $O(n)$).

To reduce the time for Case 1, we perform the following preprocessing.
For each leaf $v$ of $\Psi(\pi)$, we preprocess the path $\pi(v)$ in
the same way as above for preprocessing $\pi$. This means that we
construct an interval tree $T(\pi(v))$ as well as a decomposition tree
$\Psi(\pi(v))$ for the subpath $\pi(v)$. To answer a query for Case 1,
we instead use
$\Psi(\pi(v))$ (and use $T(\pi(v))$ as the reference tree). The query
time becomes $O(\log^2 \log n)$ as $|\pi(v)|\leq \log^2 n$. Note that
to construct $T(\pi(v))$ and $\Psi(\pi(v))$ in $O(|\pi(v)|)$ time, we
need to sort all vertices of $\pi(v)$ by $x$-coordinate in $O(|\pi(v)|)$ time. Recall
that we already have a sorted list of all vertices of $\pi$, from
which we can obtain sorted lists for $\pi(v)$ for all leaves $v$ of
$\Psi(\pi)$ in $O(n)$ time altogether.
Hence, the preprocessing for $\pi(v)$ for all leaves $v$ of
$\Psi(\pi)$ takes $O(n)$ time.

We proceed to Case 2. To reduce the query time to $O(\log n)$, we will
discuss how to perform additional preprocessing so that
$T^+(\pi(i,k))$ can be computed in $O(\log n)$ time.
Computing $T^+(\pi(k+1,j))$ can be done in $O(\log n)$ time
similarly. Finally we can merge the two trees to obtain
$T^+(\pi(i,j))$ in additional $O(\log n)$ time by Lemma~\ref{lem:merge}.

To compute $T^+(\pi(i,k))$ in $O(\log n)$ time, according
to our algorithm it suffices to compute both $T^+(i,b_v)$ and
$T^+(b_v+1,k)$ in $O(\log n)$ time. We discuss $T^+(i,b_v)$ first.

\subparagraph{Dealing with $\boldsymbol{T^+(\pi(i,b_v))}$.}
To compute $T^+(i,b_v)$ in $O(\log n)$ time, we preform the following
additional preprocessing. For each leaf $v$ of $\Psi(\pi)$, recall
that $|\pi(v)|\leq \log^2 n$;
we partition $\pi(v)$ into $t_v\leq \log n$ subpaths each of which
contains at most $\log n$ vertices. We use
$\pi_v(1),\pi_v(2),\ldots,\pi_{v}(t_v)$ to refer to these subpaths in
order along $\pi(v)$. For each subpath $\pi_v(i)$, we compute
$T^+(\pi_v(i))$ from scratch in $O(\log n)$ time. The total time for
computing all such trees is $O(\log^2 n)$. Next, we compute
compact interval trees for $t_v$ {\em prefix subpaths} of $\pi(v)$.
Specifically, for each $t\in [1,t_v]$, we compute $T^+(\pi_v[1,t])$,
where $\pi_v[1,t]$ is the concatenation of the paths
$\pi_v(1),\pi_v(2), \ldots, \pi_v(t)$. This can be done in $O(\log^2 n)$
time by computing $T^+(\pi_v[1,t])$ incrementally for $t=1,2,\ldots,
t_v$ using the merge algorithm of Lemma~\ref{lem:merge}. Indeed,
initially $T^+(\pi_v[1,t])=T^+(\pi_v(1))$, which is already available.
Then, for each $2\leq t\leq t_v$, $T^+(\pi_v[1,t])$ can be produced by
merging $T^+(\pi_v[1,t-1])$ and $T^+(\pi_v(t))$ in $O(\log n)$ time.
Similarly, we compute compact interval trees for $t_v$ {\em suffix subpaths} of $\pi(v)$:
$T^+(\pi_v[t,t_v])$ for all $t=1,2,\ldots,t_v$,
where $\pi_v[t,t_v]$ is the concatenation of the paths
$\pi_v(t),\pi_v(t+1), \ldots, \pi_v(t_v)$. This can be done in
$O(\log^2 n)$ time by a similar algorithm as above.
Thus, the preprocessing on $v$ takes $O(\log^2 n)$ time;
the preprocessing on all leaves of $\Psi(\pi)$ takes $O(n)$ time in
total.

%With the above additional preprocessing,
We can now compute
$T^+(i,b_v)$ in $O(\log n)$ time as follows. Recall that $\pi(i,b_v)$
is a subpath of $\pi(v)$ and $b_v$ is the last vertex of $\pi(v)$.
We first determine the subpath $\pi_v(t)$ that contains $i$.
Let $g$ be the last vertex of $\pi_v(t)$. We partition $\pi(i,b_v)$
into two subpaths $\pi(i,g)$ and $\pi(g+1,b_v)$, and we will compute
their compact interval trees separately and then merge them to obtain
$T^+(\pi(i,b_v))$. For $\pi(i,g)$, as $\pi(i,g)$ is a subpath of
$\pi_v(t)$ and $|\pi_v(t)|\leq \log n$, we can compute $T^+(\pi(i,g))$
from scratch in $O(\log n)$ time. For $\pi(g+1,b_v)$, observe that
$\pi(g+1,b_v)$ is exactly the suffix supath $\pi_v[t+1,t_v]$, whose
compact interval tree has already been computed in the preprocessing.
Hence, $T^+(i,b_v)$ can be produced in $O(\log n)$ time.

\subparagraph{Dealing with $\boldsymbol{T^+(\pi(b_v+1,k))}$.}
To compute $T^+(b_v+1,k)$ in $O(\log n)$ time, we perform the
following preprocessing, which was also used by Guibas et al.~\cite{ref:GuibasCo91}.
Recall that $\pi(b_v+1,k)$ is the concatenation of the canonical paths
of $O(\log n)$ nodes that are right children of the nodes on the path
in $\Psi(\pi)$ from $v$'s parent to the left child of $w$ (and these
nodes themselves are not on the path). Hence, this
sequence of nodes can be uniquely determined by the leaf-ancestor pair
$(v,w)$; we use $\pi_{v,w}$ to denote the above concatenated subpath
of $\pi$.

Correspondingly, in the preprocessing, for each leaf $v$ we do
the following. For each ancestor $w$ of $v$, we compute the compact
interval tree for the subpath $\pi_{v,w}$. As $v$ has $O(\log n)$
ancestors, computing the trees for all ancestors takes $O(\log^2 n)$
time using the merge algorithm of Lemma~\ref{lem:merge}. Hence, the
total preprocessing time on $v$ is $O(\log^2 n)$, and thus the total
preprocessing time on all leaves of $\Psi(\pi)$ is $O(n)$, for
$\Psi(\pi)$ has $O(n/\log^2 n)$ leaves. Due to the
above preprocessing, $T^+(b_v+1,k)$ is available during queries.

\subparagraph{Wrapping up.}
In summary, with $O(n)$ time preprocessing (excluding the time for
sorting the vertices of $\pi$), we can build a data structure of
$O(n)$ space that can answer each subpath hull query
in $O(\log n)$ time. Comparing with the method of Guibas et
al.~\cite{ref:GuibasCo91}, our innovation is threefold. First, we
process subpaths individually to handle queries of Case 1. Second, we
precompute the compact interval trees for convex hulls of the prefix
and suffix subpaths of $\pi(v)$ for each leaf $v$ of $\Psi(\pi)$. Third,
we use a smaller decomposition tree $\Psi(\pi)$ of only $O(n/\log^2
n)$ nodes.
The following theorem summarizes our result.

\begin{theorem}\label{theo:subpathhull}
Given a simple path $\pi$ of $n$ vertices in the plane, after all vertices
are sorted by $x$-coordinate, a data structure of $O(n)$ space can be
built in $O(n)$ time so that each subpath hull query can be
answered in $O(\log n)$ time. The query algorithm produces
a compact interval tree representing the convex hull of
the query subpath, which can support all binary-search-based
operations on the convex hull
in $O(\log n)$ time each. These operations include (but are not limited to) the following
(let $\pi'$ denote the query subpath and let $H(\pi')$ be its convex
hull):
\begin{enumerate}
	  \item
      Given a point, decide whether the point is in $H(\pi')$.
      \item
      Given a point outside $H(\pi')$, find the two tangents from the point to $H(\pi')$.
      \item
      Given a direction, find the most extreme point of $\pi'$ along the direction.
      \item
      Given a line, find its intersection with $H(\pi')$.
      \item
      Given a convex polygon (represented in any data structure that
	  supports binary search), decide whether it intersects $H(\pi')$,
	  and if not, find their common tangents (both outer and inner).
\end{enumerate}
In addition,  $H(\pi')$ can be output in time linear in the number of vertices of $H(\pi')$.
\end{theorem}
\begin{proof}
Refer to Guibas et al.~\cite{ref:GuibasCo91} for some details on how
to perform operations on the convex hull $H(\pi')$ using compact interval trees.
\end{proof}

\section{Ray-Shooting}
\label{sec:ray}

In this section, we present our results on the ray-shooting problem. The ray-shooting problem among lines is discussed in Section~\ref{sec:line}. Section~\ref{sec:segment} is concerned with the intersection detection problem and the ray-shooting problem among segments.

\subsection{Ray-shooting among lines}
\label{sec:line}

Given a set of $n$ lines in the plane, we wish to build a data structure so that the first line hit by a query ray can be found efficiently. The problem is usually tackled in the dual plane, e.g.,~\cite{ref:ChengAl92}. Let $P$ be the set of dual points of the lines. In the dual plane, the problem is equivalent to the following: Given a query line $l_q$, a pivot point $q\in l_q$, and a rotation direction (clockwise or counterclockwise), find the first point of $P$ hit by rotating $l_q$ around $q$.

A spanning path $\pi(P)$ of $P$ is a polygonal path connecting all points of $P$ such that $P$ is the vertex set of the path. Hence, $\pi(P)$ corresponds to a permutation of $P$. For any line $l$ in the plane, let $\sigma(l)$ denote the number of edges of $\pi(P)$ crossed by $l$. The {\em stabbing number} of $\pi(P)$ is the largest $\sigma(l)$ of all lines $l$ in the plane. It is known that a spanning path of $P$ with stabbing number $O(\sqrt{n})$ always exists~\cite{ref:ChazelleQu89}, which can be computed in $O(n^{1+\delta})$ time using Matou\v{s}ek's partition tree~\cite{ref:MatousekRa93} (e.g., by a method in~\cite{ref:ChazelleQu89}). Let $\pi'(P)$ denote such a path.
Note that $\pi'(P)$ may have self-intersections. Using $\pi'(P)$, Edelsbrunner et al.~\cite{ref:EdelsbrunnerIm89} gave an algorithm that can produce another spanning path $\pi(P)$ of $P$ such that the stabbing number of $\pi(P)$ is also $O(\sqrt{n})$ and $\pi(P)$ has no self-intersections (i.e., $\pi(P)$ is a simple path); the runtime of the algorithm is $O(n^{1.5})$.
Below we will use $\pi(P)$ to solve our problem.

We first build a data structure in the following lemma for $\pi(P)$.
%given by Chazelle and Guibas~\cite{ref:ChazelleFr862}.

\begin{mylemma}{\em (Chazelle and Guibas~\cite{ref:ChazelleFr862})}\label{lem:edgecross}
We can build a data structure of $O(n)$ size in $O(n\log n)$ time for any simple path of $n$ vertices, so that given any query line $l_q$, if $l_q$ intersects the path in $k$ edges, then these edges can be found in $O(k\log \frac{n}{k})$ time.
\end{mylemma}

Then, we construct the subpath hull query data structure of Theorem~\ref{theo:subpathhull} for $\pi(P)$. This finishes our preprocessing.

Given a query line $l_q$, along with the pivot $q$ and the rotation direction, we first use Lemma~\ref{lem:edgecross} to find the edges of $\pi(P)$ intersecting $l_q$. As the stabbing number of $\pi(P)$ is $O(\sqrt{n})$, this steps finds $O(\sqrt{n})$ edges intersecting $l_q$ in $O(\sqrt{n}\log n)$ time. Then, using these edges we can partition $\pi(P)$ into $O(\sqrt{n})$ subpaths each of which does not intersect $l_q$. For each subpath, we use our subpath hull query data structure to compute its convex hull in $O(\log n)$ time. Next, we compute the tangents from the pivot $q$ to each of these $O(\sqrt{n})$ convex hulls, in $O(\log n)$ time each by Theorem~\ref{theo:subpathhull}. Using these $O(\sqrt{n})$ tangents, based on the rotation direction of $l_q$, we can determine the first point of $P$ hit by $l_q$ in additional $O(\sqrt{n})$ time. Hence, the total time of the query algorithm is $O(\sqrt{n}\log n)$.

\begin{theorem}\label{theo:raylines}
There exists a data structure of complexity $O(n^{1.5},n,\sqrt{n}\log n)$ for the ray-shooting problem among lines. The preprocessing time can be reduced to $O(n\log n)$ time by a randomized algorithm while the query time is bounded by $O(\sqrt{n}\log n)$ with high probability.
\end{theorem}
\begin{proof}
We first discuss the deterministic result. The query time is $O(\sqrt{n}\log n)$, as explained above. The space is used for the data structure in Lemma~\ref{lem:edgecross} and the subpath hull query data structure in Theorem~\ref{theo:subpathhull}, which is $O(n)$. For the preprocessing time, %computing $\pi'(P)$ takes $O(n^{1+\delta})$ time.
computing $\pi(P)$ takes $O(n^{1.5})$ time. Building the data structure for Lemma~\ref{lem:edgecross} and the subpath hull query data structure can be done in $O(n\log n)$ time. Hence, the total preprocessing time is $O(n^{1.5})$.

For the randomized result, Chan~\cite{ref:ChanOp12} gave an $O(n\log n)$ time randomized algorithm to compute a spanning path $\pi''(P)$ for $P$ such that $\pi''(P)$ is a simple path and the stabbing number of $\pi''(P)$ is at most $O(\sqrt{n})$ with high probability. After having $\pi''(P)$, we build the data structure for Lemma~\ref{lem:edgecross} and the subpath hull query data structure. Hence, the preprocessing takes $O(n\log n)$ time and $O(n)$ space, and the query time is bounded by $O(\sqrt{n}\log n)$ with high probability.
%because $O(\sqrt{n})$ is the expected stabbing number of $\pi''(P)$.
\end{proof}

\subparagraph{Remark.}
As indicated in~\cite{ref:EdelsbrunnerIm89}, ray-shooting can be used to determine whether two query points $p$ and $q$ are in the same face of the arrangement of a set of lines. Indeed, let $\rho$ be the ray originated from $p$ towards $q$. Then, $p$ and $q$ are in the same face of the arrangement if and only if $\rho$ hits the first line after $q$.

\medskip

We can extend the above algorithm to obtain the following result on the first-$k$-hit queries.

\begin{theorem}\label{theo:khitraylines}
Given a set of $n$ lines in the plane, we can build a data structure of $O(n)$ space in $O(n^{1.5})$ time so that given a ray and an integer $k$, we can find the first $k$ lines hit by the ray in $O(\sqrt{n}\log n + k\log n)$ time. The preprocessing time can be reduced to $O(n\log n)$ while the query time is bounded by $O(\sqrt{n}\log n+k\log n)$ with high probability.
\end{theorem}
\begin{proof}
We still work in the dual plane and use the same notation as above. In the dual plane, the problem is equivalent to finding the first $k$ points that are hit by $l_q$ when it is rotating around the pivot $q$ following the given direction.
We perform exactly the same processing as before.
Let $p_1,p_2,\ldots,p_n$ be the points of $P$ ordered along $\pi(P)$.

\begin{figure}[t]
\begin{minipage}[t]{\textwidth}
\begin{center}
\includegraphics[height=1.0in]{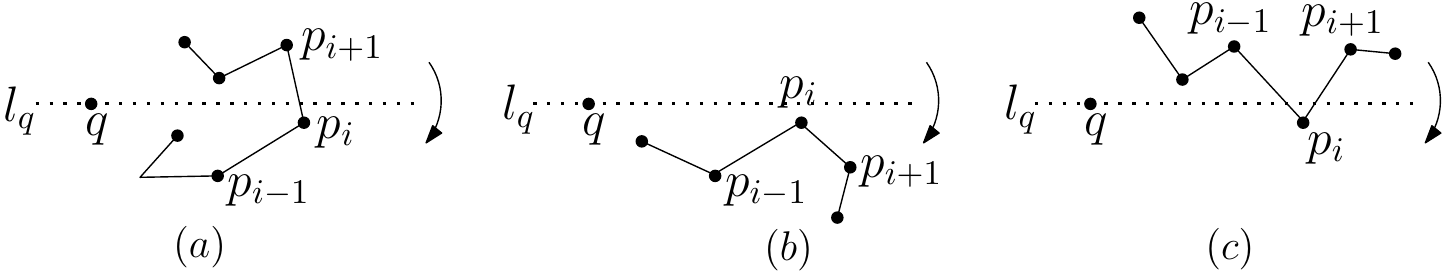}
\caption{\footnotesize Illustrating the three cases for Theorem~\ref{theo:khitraylines}: $l_q$ is the dashed line rotating clockwise around $q$.}
\label{fig:khit}
\end{center}
\end{minipage}
\vspace{-0.15in}
\end{figure}

Consider a query with $l_q$ and $q$. We first determine a set $\Pi$ of $O(\sqrt{n})$ subpaths of $\pi(P)$ that do not intersect $l_q$. Then, we find the first point $p_i$ hit by rotating $l_q$ in the same way as before. This takes $O(\sqrt{n}\log n)$ time. We continue rotating $l_q$ to find the second point. To this end, we need to update the set $\Pi$ so that the new $\Pi$ contains the $O(\sqrt{n})$ subpaths of $\pi(P)$ that do not intersect $l_q$ at its current position (i.e., after it rotated over $p_i$). As $l_q$ has rotated over only one point of $P$, we can update $\Pi$ in constant time as follows.

If $p_{i-1}$ and $p_{i+1}$ are in different sides of $l_q$, then $p_i$ is an endpoint of a subpath $\pi_1$ of $\Pi$ (e.g., see Fig.~\ref{fig:khit}(a)). Without loss of generality, we assume that $p_{i-1}$ is also in $\pi_1$. Thus, $p_{i+1}$ is the endpoint of another subpath $\pi_2$. To update $\Pi$, we remove $p_i$ from $\pi_1$ and append $p_i$ to $\pi_2$ (so $p_i$ becomes a new endpoint of $\pi_2$).

If $p_{i-1}$ and $p_{i+1}$ are in the same side of $l_q$, then there are two subcases depending on whether $p_i$ and $p_{i-1}$ are in the same side of $l_q$, where $l_q$ refers to the line at its original position before it rotated over $p_i$.
If $p_i$ and $p_{i-1}$ are in the same side of $l_q$, then all three points $p_{i-1},p_i,p_{i+1}$ are in the same subpath $\pi_3$ of $\Pi$ (e.g., see Fig.~\ref{fig:khit}(b)). To update $\Pi$, we break $\pi_3$ into three subpaths by removing the two edges $\overline{p_{i-1}p_i}$ and $\overline{p_{i}p_{i+1}}$ (so $p_i$ itself forms a subpath). If $p_i$ and $p_{i-1}$ are not in the same side of $l_q$, then the three points $p_{i-1},p_i,p_{i+1}$ are in three different subpaths of $\Pi$ (in particular, $p_i$ itself forms a subpath;  e.g., see Fig.~\ref{fig:khit}(c)). To update $\Pi$, we merge these three subpaths into one subpath.

Since updating $\Pi$ only involves $O(1)$ subpath changes as discussed above, we can compute the convex hulls of the new subpaths and the tangents from $q$ in $O(\log n)$ time by Theorem~\ref{theo:subpathhull}. Hence, computing the next hit point takes $O(\log n)$ time. We continue rotating $l_q$ in this way until $k$ points are found. The total query time is bounded by $O(\sqrt{n}\log n+k\log n)$.

For the same reason as in Theorem~\ref{theo:raylines}, the randomized result also follows.
\end{proof}

\subsection{Intersection detection and ray-shooting among segments}
\label{sec:segment}

Given a set $S$ of $n$ segments in the plane, an intersection detection query asks whether a query line intersects at least one segment of $S$. One motivation to study the problem is that it is a subproblem in our algorithm for the ray-shooting problem among segments.

%needs to solve a special case of this problem where the query ray is a line. However, our techniques can be easily generalized to solving the more general case. Hence, here we consider the rays as query objects.

To find a data structure to store the segments of $S$, we adapt the techniques of Overmars et al.~\cite{ref:OvermarsSt90} to the partition trees of Matou\v{s}ek~\cite{ref:MatousekEf92,ref:MatousekRa93} (to obtain the deterministic result) as well as that of Chan~\cite{ref:ChanOp12} (to obtain the randomized result). To store segments, Overmars et al.~\cite{ref:OvermarsSt90} used a so-called {\em interval partition tree}, whose underling structure is a conjugation tree of Edelsbrunner and Welzl~\cite{ref:EdelsbrunnerHa86}. The idea is quite natural due to the nice properties of conjugation trees: Each parent region is partitioned into exactly two disjoint children regions by a line. The drawback of conjugation trees is the slow $\widetilde{O}(n^{0.695})$ query time.
When adapting the techniques to more query-efficient partition trees such as those in \cite{ref:ChanOp12,ref:MatousekEf92,ref:MatousekRa93}, two issues arise. First, each parent region may have more than two children. Second, children regions may overlap. Chan's partition tree~\cite{ref:ChanOp12} does not have the second issue while both issues appear in Matou\v{s}ek's partition trees~\cite{ref:MatousekEf92,ref:MatousekRa93}. As a matter of fact, the second issue incurs a much bigger challenge. In the following, we first present our randomized result by using Chan's partition tree~\cite{ref:ChanOp12}, which is relatively easy, and then discuss the deterministic result using Matou\v{s}ek's partition trees~\cite{ref:MatousekEf92,ref:MatousekRa93}. The description of the randomized result may also serve as a ``warm-up'' for our more complicated deterministic result.

We begin with the following lemma, which solves a special case of the problem. The lemma will be needed in both our randomized and deterministic results.

\begin{mylemma}\label{lem:anchor}
Suppose all segments of $S$ intersect a given line segment.
\begin{enumerate}
\item
We can build a data structure of $O(n)$ space in $O(n\log n)$ time so that whether a query line intersects any segment of $S$ can be determined in $O(\log n)$ time.
\item
If the segments of $S$ are nonintersecting, then we can build a data structure of $O(n)$ space in $O(n\log n)$ time so that the first segment hit by a query ray can be found in $O(\log n)$ time.
\end{enumerate}
\end{mylemma}
\begin{proof}
Let $s$ be the line segment that intersects all segments of $S$.
Without loss of generality, we assume that $s$ is horizontal. Let $\ell$ be the line containing $s$. For each segment $s'\in S$, we divide it into two subsegments by its intersection with $\ell$; let $S_1$ (resp., $S_2$) be the set of all such subsegments above (resp., below) $\ell$. In the following we describe our preprocessing algorithm for $S_1$; the set $S_2$ will be preprocessed by the same algorithm.

We consider the line segment arrangement $\calA$ of all segments of $S_1$ and the line $\ell$ in the closed halfplane above $\ell$. Alevizos et al.~\cite{ref:AlevizosAn90} proved that every cell of $\calA$ is of complexity $O(n)$. Let $C$ denote the external cell of $\calA$, i.e., the cell containing the left endpoint point of $s$. Alevizos et al.~\cite{ref:AlevizosAn90} gave an $O(n\log n)$ time algorithm to compute $C$. As $C$ is simply connected, we may treat it as a simple polygon; for this, we could add two edges at infinity so that the closed halfplane above $\ell$ becomes a big triangle and we call the two edges {\em dummy edges}. In $O(n)$ time we build a point location data structure~\cite{ref:EdelsbrunnerOp86,ref:KirkpatrickOp83} on $C$ so that given any point $p$ in the plane, we can determine whether $p\in C$ in $O(\log n)$ time. We also build a ray-shooting data structure~\cite{ref:ChazelleRa94,ref:ChazelleVi89,ref:HershbergerA95} on $C$ in $O(n)$ time so that given a ray whose origin is in $C$, the first edge of the boundary $\partial C$ hit by the ray can be found in $O(\log n)$ time. This finishes our preprocessing for $S_1$, which uses $O(n\log n)$ time and $O(n)$ space.
We do the same preprocessing for $S_2$.

Given a query line $l$, $l$ intersects a segment of $S$ if and only if it intersects a segment of $S_1\cup S_2$. Hence, it suffices to determine whether $l$ intersects a segment of $S_1$ and whether $l$ intersects a segment of $S_2$. Below we show that whether $l$ intersects a segment of $S_1$ can be determined in $O(\log n)$ time. The same is true for the case of $S_2$.

We first assume that $l$ is not parallel to $\ell$. Let $p$ be the intersection of $l$ and $\ell$. We first determine whether $p$ is in $C$ by the point location data structure on $C$. If $p\not\in C$, then $p$ is in an internal cell of $\calA$, implying that $l$ must intersect a segment of $S_1$. Otherwise, let $\rho$ be the ray from $p$ going upwards. Using the ray-shooting data structure, we find the first edge $e$ of $\partial C$ hit by $\rho$. Observe that $l$ intersects a segment of $S_1$ if and only if $e$ is not a dummy edge. Hence, we can determine whether $l$ intersects a segment of $S_1$ in $O(\log n)$ time. If $l$ is parallel to $\ell$, then we can use a similar algorithm.
%Now suppose that $l$ is parallel to $\ell$. If $l$ is below $\ell$, then $l$ does not intersect any segment of $S_1$. Otherwise, $l$ must intersect both dummy edges. Let $p$ be the intersection
This proves the first statement of the lemma.

\medskip
For the second statement of the lemma, since the segments of $S$ are nonintersecting, $C$ is the only cell of $\calA$. This nice property can help us to answer the ray-shooting problem on $S$. We build a ray-shooting data structure on $C$ as above. We do the same preprocessing for $S_2$.

Given any query ray $\rho$ with origin $p$.
To find the first segment of $S$ hit by $\rho$, it is sufficient to find the first segment of $S_1$ hit by $\rho$ and find the first segment of $S_2$ hit by $\rho$. In the following, we show that the first segment of $S_1$ hit by $\rho$ can be found in $O(\log n)$ time. The same algorithm works for the case $S_2$ as well.

Without loss of generality, we assume that $\rho$ is going upwards.
If $p$ is above $\ell$, then $p$ is in $C$. Using the ray-shooting data structure, we find the first edge $e$ of $\partial C$ hit by $\rho$. If $e$ is a dummy edge, then $\rho$ does not hit any segment of $S_1$; otherwise, the segment that contains $e$ is the first segment of $S_1$ hit by $\rho$. If $p$ is below $\ell$, let $p'$ be the intersection between $\rho$ and $\ell$. Now we can follow the same algorithm as above by considering $p'$ as the new origin of $\rho$. Hence, the query time is $O(\log n)$.
\end{proof}

\subsubsection{The randomized result}
\label{sec:randomized}

We first briefly review Chan's partition tree~\cite{ref:ChanOp12} (which works for any fixed dimensional space; but for simplicity we only discuss it in 2D, which suffices for our problem). Chan's partition tree for a set $P$ of $n$ points, denoted by $T$, is a hierarchical structure by recursively subdividing the plane into triangles. Each node $v$ of $T$ corresponds to a triangle, denoted by $\triangle(v)$. If $v$ is the root, then $\triangle(v)$ is the entire plane. If $v$ is not a leaf, then $v$ has $O(1)$ children whose triangles form a disjoint partition of $\triangle(v)$. Define $P(v)=P\cap \triangle(v)$. The set $P(v)$ is not explicitly stored at $v$ unless $v$ is a leaf, in which case $|P(v)|=O(1)$. The height of $T$ is $O(\log n)$. Let $\kappa(T)$ denote the maximum number of triangles of $T$ that are crossed by any line in the plane. Chan~\cite{ref:ChanOp12} gave an $O(n\log n)$ time randomized algorithm to compute $T$ such that $\kappa(T)$ is at most $O(\sqrt{n})$ with high probability.

Let $P$ be the set of the endpoints of all segments of $S$ (so $|P|=2n$). We first build the tree $T$ as above.
We then store the segments of $S$ in $T$, as follows. For each segment $s$, we apply the following algorithm. Starting from the root of $T$, for each node $v$, we assume that $s$ is contained in $\triangle(v)$, which is true when $v$ is the root. If $v$ is a leaf, then we store $s$ at $v$; let $S(v)$ denote all segments stored at $v$. If $v$ is not a leaf, then we check whether $s$ is in $\triangle(u)$ for a child $u$ of $v$. If yes, we proceed on $u$. Otherwise, for each child $u$, for each edge $e$ of $\triangle(u)$, if $s$ intersects $e$, then we store $s$ at the edge $e$ (in this case we do not proceed to the children of $u$); denote by $S(e)$ the set of edges stored at $e$. This finishes the algorithm for storing $s$. As each node of $T$ has $O(1)$ children, $s$ is stored $O(1)$ times and the algorithm runs in $O(\log n)$ time. In this way, it takes $O(n\log n)$ time to store all segments of $S$, and the total sum of $|S(e)|$ and $|S(v)|$ for all triangle edges $e$ and all leaves $v$ is $O(n)$. In addition, $|S(v)|=O(1)$ for any leaf $v$, since $|P(v)|=O(1)$ and both endpoints of each segment $s\in S(v)$ are in $P(v)$.

Next, for each triangle edge $e$, since all edges of $S(e)$ intersect $e$,
we preprocess $S(e)$ using Lemma~\ref{lem:anchor}(1). Doing this for all triangle edges $e$ takes $O(n\log n)$ time and $O(n)$ space.
%as the total sum of $|S(e)|$ for all triangle edges $e$ is $O(n)$.
%This finishes our preprocessing, whose overall time is $O(n\log n)$.

Consider a query line $l$. Our goal is to determine whether $l$ intersects any segment of $S$. Starting from the root, we determine the set of nodes $v$ whose triangles $\triangle(v)$ are crossed by $l$.
For each such node $v$, if $v$ is a leaf, then we check whether $s$ intersects $l$ for each segment $s\in S(v)$; otherwise, for each edge $e$ of $\triangle(v)$, we use the query algorithm of Lemma~\ref{lem:anchor}(1) to determine whether $l$ intersects any segment of $S(e)$. As the number of nodes $v$ whose triangles $\triangle(v)$ crossed by $l$ is at most $\kappa(T)$ and $S(v)=O(1)$ for each leaf $v$, the total time of the query algorithm is $O(\kappa(T)\cdot \log n)$. The correctness of the algorithm is discussed in the proof of Theorem~\ref{theo:intersectionrandomize}.

\begin{theorem}\label{theo:intersectionrandomize}
Given a set $S$ of $n$ (possibly intersecting) segments in the plane, we can build a data structure of $O(n)$ space in $O(n\log n)$ time so that whether a query line intersects any segment of $S$ can be determined in $O(\sqrt{n}\log n)$ time with high probability.
\end{theorem}
\begin{proof}
We have discussed the preprocessing time and space. We have also shown that the query time is  $O(\kappa(T)\cdot \log n)$. Since $\kappa(T)$ is bounded by $O(\sqrt{n})$ with high probability, the query time is bounded by $O(\sqrt{n}\log n)$ with high probability. It remains to show the correctness of the query algorithm. Indeed, if the algorithm reports the existence of an intersection, then according to our algorithm, it is true that $l$ intersects a segment of $S$. On the other hand, suppose $l$ intersects a segment $s$, say, at a point $p$. If $s$ is stored at $S(v)$ for a leaf $v$, then $l$ must cross $\triangle(v)$ and thus our algorithm will detect the intersection. Otherwise, $s$ must be stored in $S(e)$ for an edge $e$ of a triangle $\triangle(u)$ that contains $p$. Since $p\in l$, $l$ must cross $\triangle(u)$. According to our query algorithm, the query algorithm of Lemma~\ref{lem:anchor}(1) will be invoked on $S(e)$, and thus the algorithm will report the existence of intersection.
\end{proof}

Suppose the segments of $S$ are nonintersecting.
In the above algorithm, if we replace Lemma~\ref{lem:anchor}(1) by Lemma~\ref{lem:anchor}(2) in both the preprocessing and query algorithms, then we can obtain the following result.

\begin{theorem}
Given a set $S$ of $n$ nonintersecting segments in the plane, we can build a data structure of $O(n)$ space in $O(n\log n)$ time so that the first segment of $S$ hit by a query ray can be found in $O(\sqrt{n}\log n)$ time  with high probability.
\end{theorem}
\begin{proof}
In the preprocessing, we use Lemma~\ref{lem:anchor}(2) to preprocess $S(e)$ for each triangle edge $e$. The total preprocessing time is $O(n\log n)$ and the space is $O(n)$. Given a query ray $\rho$, we find the set of nodes $v$ whose triangles $\triangle(v)$ are crossed by $l$ in $O(\kappa(T))$ time. For each such node $v$, if $v$ is a leaf, then we check whether $\rho$ hits $s$ for each segment $s\in S(v)$. Otherwise, for each edge $e$ of $\triangle(v)$, we use the query algorithm of Lemma~\ref{lem:anchor}(2) to find the first segment of $S(e)$ hit by $\rho$. Finally, among all segments found above that are hit by $\rho$, we return the one whose intersection with $\rho$ is closest to the origin of $\rho$. The time analysis and algorithm correctness are similar to those of Theorem~\ref{theo:intersectionrandomize}.
\end{proof}

To solve the ray-shooting problem among (possibly intersecting) segments, as discussed in Section~\ref{sec:introray}, Cheng and Janardan~\cite{ref:ChengAl92} gave an algorithm that uses both an algorithm for the ray-shooting problem among lines and an algorithm for the intersection detection problem. If we replace their algorithms for these two problems by our new results in Theorems~\ref{theo:raylines} and~\ref{theo:intersectionrandomize}, then we can obtain Theorem~\ref{theo:randomray}. For the completeness of this paper, we reproduce Cheng and Janardan's algorithm~\cite{ref:ChengAl92} in the proof of Theorem~\ref{theo:randomray}.
%(refer to \cite{ref:ChengAl92} for details of their algorithm).

\begin{theorem}\label{theo:randomray}
Given a set $S$ of $n$ (possibly intersecting) segments in the plane, we can build a data structure of $O(n\log n)$ space in $O(n\log^2 n)$ time such that the first segment of $S$ hit by a query ray can be found in $O(\sqrt{n}\log n)$ time with high probability.
\end{theorem}
\begin{proof}
We reproduce Cheng and Janardan's data structure~\cite{ref:ChengAl92} but instead use our new results for the ray-shooting problem among lines and the intersection detection problem.

For ease of discussion, we assume that no segment of $S$ is vertical.
The underling structure is a segment tree $T$ on the segments of
$S$~\cite{ref:deBergCo08}. Specifically, let $x_1,x_2,\ldots,x_{2n}$
be the $x$-coordinates of the endpoints of the segments of $S$ sorted
from left to right. These values partition the $x$-axis into $4n+1$
intervals $(-\infty,x_1),[x_1,x_1], (x_1,x_2),[x_2,x_2],\ldots,(x_{2n},+\infty)$.
$T$ is a complete binary tree whose leaves correspond to
the above intervals in order from left to right. Each internal
node $v$ is associated with an interval $Int(v)$ that is the union of all intervals in
the leaves of $T(v)$, where $T(v)$ is the subtree rooted at $v$.
For each segment $s\in S$, it is stored at a node $v$ if $Int(v)\subseteq
[x(s),x'(s)]$ and $Int(parent(v))\not\subseteq [x(s),x'(s)]$, where
$x(s)$ and $x'(s)$ are the $x$-coordinates of the left and right endpoints of $s$,
respectively, and $parent(v)$ is the parent of $v$ in $T$; let $S(v)$ denote the set of
all segments stored at $v$.
Each segment of $s$ is stored in $O(\log n)$ nodes and the total space
is $O(n\log n)$.

The above describes a standard segment tree. For solving our
problem, each internal node $v$ also stores another set
$S'(v)=\bigcup_{u\in T(v)}S(u)$.
%In this way, each segment of $s$ is also stored in $O(\log n)$ nodes and the total space is still $O(n\log n)$.
One can check that both $|S(v)|$ and $|S'(v)|$ are bounded by
$O(|T_v|)$, where $|T_v|$ refers to the number of leaves of $T_v$. Finally,
we trim the segments of $S'(v)$ by only keeping the portions in the
vertical strip $Int(v)\times(-\infty,+\infty)$, i.e., for each segment
$s\in S'(v)$, we only keep its subsegment in the trip in $S'(v)$.

For each node $v\in T$, we construct the
ray-shooting-among-line data structure in Theorem~\ref{theo:raylines}
(using the randomized result with $O(n\log n)$ preprocessing time) on
the supporting lines of the segments of $S(v)$; let $R(v)$ denote the
data structure. We also construct the intersection detection
data structure in Theorem~\ref{theo:intersectionrandomize} on the
segments of $S'(v)$; let $D(v)$
denote the data structure.
This finishes the preprocessing for our problem, which uses $O(n\log^2 n)$ time and
$O(n\log n)$ space. We discuss the query algorithm below.

Consider a query ray $\rho_q$, with origin $q$. Without loss of
generality, we assume that $\rho_q$ goes rightwards. Starting from the
root of $T$, we locate the leaf whose interval contains $q$. Then,
from the leaf we go upwards in $T$ until we find the first node whose
right node $u$ is not on the path and $\rho_q$ intersects a segment of
$S'(u)$. Note that since segments of $S'(u)$ are all in the strip
$Int(u)\times(-\infty,+\infty)$ and $q$ is to the left of the strip
(and thus $\rho_q$ spans the strip), determining whether $\rho_q$
intersects a segment of $S'(u)$ is equivalent to determining whether
the supporting line of $\rho_q$ intersects a segment of $S'(u)$, and
thus we can use the data structure $D(u)$.
We call the above the {\em percolate-up} procedure. Next,
starting from $u$, we run a {\em percolate-down} procedure as follows. Suppose the
procedure is now considering a node $v$ (initially $v=u$). We first
find the first segment (if exists) of $S(v)$ hit by $\rho_q$ within
the strip $Int(v)\times(-\infty,+\infty)$. Notice that all segments of
$S(v)$ span the strip. Thus, the above problem can be solved by
calling the ray-shooting data structure $R(v)$
using the portion $\rho'$ of $\rho_q$ that lies to the right of
the left vertical line of the strip. We keep the segment found
by $R(v)$ if and only if the intersection of the segment and $\rho'$
is in the trip. Let $left(v)$ and $right(v)$ denote the left and right
children of $v$, respectively. Next, we check whether $\rho_q$ intersects a
segment of $S'(left(v))$, which, as discussed above, can be done by
using the data structure $D(left(v))$. If yes, then we
proceed on $left(v)$ recursively. Otherwise, we check whether $\rho_q$ intersects a
segment of $S'(right(v))$ by using the data structure $D(right(v))$.
If yes, then we proceed on $right(v)$ recursively. Otherwise, we
stop the algorithm. After the percolate-down procedure, among the segments found above (by
$R(v)$), the one whose intersection with $\rho_q$ is closest to the
origin $q$ is
the first segment of $S$ hit by $\rho_q$.

For the query time, it is not difficult to see that the percolate-up
procedure calls the intersection detection data structure $D(v)$ for
$O(\log n)$ nodes $v$, each taking $O(\sqrt{|S'(v)|}\log{n})$ time with high probability. Notice that these nodes $v$ are on distinct levels of $T$.
Recall that $|S'(v)|=O(|T_v|)$. Hence, $|S'(v)|$ decreases
geometrically if we order these nodes $v$ by their distances from the
root. Therefore, the total time on calling
$D(v)$ for all nodes $v$ is $O(\sqrt{n}\log n)$ with high probability\footnote{We provide some explanations here. Suppose calling $D(v)$ for each node $v$ takes $O(\sqrt{|S'(v)|}\log{n})$ time with probability at least $1-1/n^c$ for a constant $c$. Let $c'>0$ be a constant smaller than $c$. Then, $n^c>n^{c'}\cdot O(\log n)$ for sufficiently large $n$. Hence, calling $D(v)$ for all $O(\log n)$ nodes $v$ takes $O(\sqrt{n}\log n)$ time with probability at least $1-1/n^c\cdot O(\log n)> 1-1/n^{c'}$. Therefore, the $O(\sqrt{n}\log n)$ time bound holds with high probability.}. The percolate-down procedure calls $D(v)$ for $O(\log n)$ nodes $v$,
and at most two such nodes are at the same level of $T$. Hence,
the total time is also $O(\sqrt{n}\log n)$ with high probability. The procedure also calls the ray-shooting data
structure $R(v)$ for $O(\log n)$ nodes $v$ at distinct levels of $T$.
We also have $|S(v)|=O(|T_v|)$. Therefore, the total time of
the ray-shooting queries is $O(\sqrt{n}\log n)$ with high probability. In summary, the query
algorithm runs in $O(\sqrt{n}\log n)$ time with high probability.

\subparagraph{Remark.} Later we will present our deterministic result
for the segment detection problem with complexity
$O(n^{1.5},n,\sqrt{n}\log n)$ in Theorem~\ref{theo:segmentray}.
Using the above algorithm and
our deterministic result of the ray-shooting-among-line problem
in Theorem~\ref{theo:raylines}, we can obtain our
deterministic result for the ray-shooting-among-segment problem.
The space is $O(n\log n)$ and the query time is
$O(\sqrt{n}\log n)$, following the same analysis as above. The
preprocessing time satisfies the recurrence relation:
$T(n)=2T(n/2)+O(n^{1.5})$,
as both $|S(v)|$ and $|S'(v)|$ are bounded by $O(|T_v|)$.
Solving the recurrence relation gives $T(n)=O(n^{1.5})$.
\end{proof}

\subsubsection{The deterministic result}

To obtain the deterministic result, we turn to Matou\v{s}ek's partition trees~\cite{ref:MatousekEf92,ref:MatousekRa93}. As discussed before, a big issue is that the triangles of these trees may overlap. To overcome the issue, we have to somehow modify Matou\v{s}ek's original algorithms. %The issue is also the cause why we have one more $\sqrt{\log n}$ factor in the query time.

\subparagraph{An overview.}
To solve the simplex range searching problem (e.g., the counting problem), Matou\v{s}ek built a partition tree in~\cite{ref:MatousekEf92} with complexity $O(n\log n, n, \sqrt{n}(\log n)^{O(1)})$; subsequently, he presented a more query-efficient result in~\cite{ref:MatousekRa93} with complexity $O(n^{1+\delta}, n, \sqrt{n})$. Ideally, we want to use his second approach. In order to achieve the $O(n^{1+\delta})$ preprocessing time, Matou\v{s}ek used multilevel data structures (called partial simplex decomposition scheme in~\cite{ref:MatousekRa93}). In our problem, however, the multilevel data structures do not work any more because they do not provide a ``nice'' way to store the segments of $S$. Without using multilevel data structures, the preprocessing time would be too high (indeed Matou\v{s}ek~\cite{ref:MatousekRa93} gave a {\em basic algorithm} without using multilevel data structures but he only showed that its runtime is polynomial). By a careful implementation, we can bound the preprocessing time by $O(n^2)$.
To improve it, we resort to the {\em simplicial partition} in~\cite{ref:MatousekEf92}. Roughly speaking, let $P$ be the set of endpoints of the segments of $S$; we partition $P$ into $r=\Theta(\sqrt{n})$ subsets of size $\sqrt{n}$ each, using $r$ triangles such that any line in the plane only crosses $O(\sqrt{r})$ triangles. Then, for each subset, we apply the algorithm of~\cite{ref:MatousekRa93}. This guarantees the $O(n^{1.5})$ upper bound on the preprocessing time for all subsets. To compute the simplicial partition,  Matou\v{s}ek~\cite{ref:MatousekEf92} first provided a {\em basic algorithm} of polynomial time and then used other techniques to reduce the time to $O(n\log n)$. For our purpose, these techniques are not suitable (for a similar reason to multilevel data structures). Hence, we can only use the basic algorithm, whose time complexity is only shown to be polynomial in~\cite{ref:MatousekEf92}. Further, we cannot directly use the algorithm because the produced triangles may overlap (the algorithm in~\cite{ref:MatousekRa93} has the same issue). Nevertheless, we manage to modify the algorithm and bound its time complexity by $O(n^{1.5})$. Also, even with the above modification that avoids certain triangle overlap, using the approach in~\cite{ref:MatousekRa93} directly still cannot lead to an $O(\sqrt{n}\log n)$ time query algorithm. Instead we have to further modify the algorithm (e.g., choose a different weight function).

In the following, we first describe our algorithm for computing the simplicial partition and then preprocess each subset in the partition by modifying Matou\v{s}ek's basic algorithm in~\cite{ref:MatousekRa93}. The algorithms in \cite{ref:MatousekEf92,ref:MatousekRa93} are both for any fixed dimensions. To simplify the description, we will discuss the planar case only. For ease of reference, we start a new section.

\subsubsection{Computing a simplicial partition}
\label{sec:simpartition}

We first review some concepts.
A {\em cutting} is a set of interior-disjoint triangles whose union is
the entire plane; its {\em size} is defined to be the number of triangles.
Let $H$ be a set of $n$ lines and $\Xi$ be a cutting. For a triangle
$\triangle\in\Xi$, let $H_{\triangle}$ denote the subset of lines of
$H$ intersecting the interior of $\triangle$. We say that $\Xi$ is an
{\em $\epsilon$-cutting} for $H$ if $|H_{\triangle}|\leq \epsilon
\cdot n$ for each triangle $\triangle\in \Xi$. We also need to handle
the weighted case where each line $l$ of $H$ has a weight $w(l)$, which is
a positive integer. We use $(H,w)$ to denote the weighted line set.
For each subset $H'\subseteq H$, define $w(H')=\sum_{l\in H'}w(l)$. A
cutting $\Xi$ is an {\em $\epsilon$-cutting} for $(H,w)$ if
$w(H_{\triangle})\leq \epsilon \cdot w(H)$ for every triangle
$\triangle\in \Xi$.

\begin{mylemma}\label{lem:cutting}{\em \cite{ref:ChazelleCu93,ref:MatousekCu91}}
Given a set of $n$ weighted lines $(H,w)$, for any parameter $r\leq n$, a $(1/r)$-cutting of size $O(r^2)$ can be computed in $O(nr)$ time.
\end{mylemma}

Recall that $P$ is the set of the endpoints of $S$ and $|S|=n$. To
simplify the notation, we let $|P|=n$ in the following (and
thus $|S|=n/2$).

A {\em simplicial partion} of size $m$ for $P$ is a collection
$\Pi=\{(P_1,\triangle_1),\ldots,(P_m,\triangle_m)\}$ with the
following properties: (1) The subsets $P_i$'s form a disjoint
partition of $P$; (2) each $\triangle_i$ is an open triangle containing $P_i$; (3) $\max_{1\leq i\leq m}|P_i|\leq 2\cdot \min_{1\leq i\leq m}|P_i|$; (4) the triangles may overlap and a triangle $\triangle_i$ may contain points in $P\setminus P_i$. We define the {\em crossing number} of $\Pi$ as the largest number of triangles that are intersected by any line in the plane.

\begin{mylemma}\label{lem:simpartition}{\em\cite{ref:MatousekEf92}}
For any integer $z$ with $2\leq z<|P|$, there exists a simplicial partition $\Pi$ of size $\Theta(r)$ for $P$, whose subsets $P_i$'s satisfy $z\leq |P_i|<2z$, and whose crossing number is $O(\sqrt{r})$, where $r=|P|/z$.
\end{mylemma}

To compute such a simplicial partition as in
Lemma~\ref{lem:simpartition}, Matou\v{s}ek~\cite{ref:MatousekEf92}
first presented a basic algorithm whose runtime is polynomial and
then improved the time to $O(n\log n)$ by other techniques. As
discussed before, the techniques are not suitable for our purpose and
we can only use the basic algorithm. In addition, the above
property~(4) prevents us from using the partition directly. Instead we
use an {\em enhanced simplicial partition} with the following
modified/changed properties. In property (2), each $\triangle_i$ is
either a triangle or a convex quadrilateral; we now call
$\triangle_i$ a {\em cell}. In property (4), the cells may still
overlap, and a cell $\triangle_i$ may still contain points in
$P\setminus P_i$; however,
if $\triangle_i$ contains a point $p\in P_j$ with $j\neq i$, then all
points of $P_i$ are outside $\triangle_j$ (e.g., see Fig.~\ref{fig:overlap}). This modified property (4), which we call {\em the weakly-overlapped property}, is the key to guarantee the success of our approach. We use convex
	quadrilaterals instead of only triangles to make sure that
	the modified property (4) can be achieved. The crossing number of
	the enhanced partition is defined as the largest number of cells
	that are intersected by any line in the plane.
We will show that by modifying Matou\v{s}ek's basic algorithm~\cite{ref:MatousekEf92}, we can compute an enhanced simplicial partition with the same feature as Lemma~\ref{lem:simpartition}. Roughly speaking, each cell of our partition is a subset of a triangle of the partition computed by  Matou\v{s}ek's algorithm. For our purpose, we are interested in the parameters $z=\sqrt{n}$ and thus $r=\Theta(\sqrt{n})$. We will show that such an enhanced simplicial partition with crossing number $O(\sqrt{r})$ can be computed in $O(n^{1.5})$ time. To this end, we first review Matou\v{s}ek's basic algorithm~\cite{ref:MatousekEf92}.
Below we fix $r=\sqrt{n}$ (and thus $z=n/r=\sqrt{n}$).

\begin{figure}[t]
\begin{minipage}[t]{\textwidth}
\begin{center}
\includegraphics[height=1.5in]{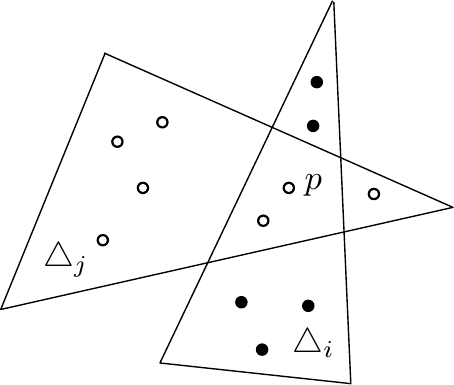}
\caption{\footnotesize Illustrating the weakly-overlapped property: $P_j$ consists of all circle points and $P_i$ consists of all disk points. A point $p\in P_j$ is also contained in $\triangle_i$, but all points of $P_i$ are outside $\triangle_j$.
}
\label{fig:overlap}
\end{center}
\end{minipage}
\vspace{-0.15in}
\end{figure}

The first main step is to compute a {\em test set} $H$ of $r$ lines (i.e., Lemma 3.3 of~\cite{ref:MatousekEf92}). This is done by computing a $(1/t)$-cutting $\Xi$ for the dual lines of the points of $P$ such that $\Xi$ has at most $r$ vertices in total, where $t$ can be chosen so that $t=\Theta(\sqrt{r})$. The set $H$ is just the dual lines in the primal plane of the vertices of $\Xi$. By Lemma~\ref{lem:cutting}, this step can be done in $O(n\sqrt{r})$ time.

The second main step is to construct the simplical partition $\Pi$ by using $H$ (i.e., Lemma 3.2 of~\cite{ref:MatousekEf92}). The algorithm has $m$ iterations and the $i$-th iteration will compute the pair $(P_i,\triangle_i)$, for $1\leq i\leq m$, with $m=\Theta(r)$. Suppose that $(P_1,\triangle_1),\ldots,(P_i,\triangle_i)$ have been computed. Let $P_i'=P\setminus(P_1\cup \cdots \cup P_i)$ and $n_i=|P_i'|$. The algorithm for computing $(P_{i+1},\triangle_{i+1})$ works as follows. If $n_i<2z$, then set $P_{i+1}=P_i'$ and set $\triangle_{i+1}$ to be the whole plane, which finishes the entire algorithm.
We next discuss the case $n_i\geq 2z$.

We define a weighted line set $(H,w_i)$: For each line $l\in H$, define $w_i(l)=2^{k_i(l)}$, where $k_i(l)$ is the number of triangles among $\triangle_1,\ldots,\triangle_i$ crossed by $l$. We compute a $(1/t_i)$-cutting $\Xi_i$ for $(H,w_i)$ for a largest possible value $t_i$ such that $\Xi_i$ has at most $n_i/z$ triangles. By Lemma~\ref{lem:cutting}, we can choose $t_i$ such that $t_i=\Theta(\sqrt{n_i/z})$. As $\Xi_i$ has at most $n_i/z$ triangles, it has a triangle that contains at least $z$ points of $P_i'$. Let $\triangle_{i+1}$ be such a triangle and choose any $z$ points of $P_i'\cap \triangle_{i+1}$ to constitute $P_{i+1}$. This finishes the construction of $(P_{i+1},\triangle_{i+1})$.

Matou\v{s}ek~\cite{ref:MatousekEf92} proved that the crossing number of $\Pi$ thus constructed is $O(\sqrt{r})$.

To compute our enhanced simplical partition, we slightly modify the above algorithm as follows (we only point out the changes). In the case $n_i\geq 2z$, let $\triangle$ be a triangle of $\Xi_i$ that contains at least $z$ points of $P_i'$. Let $\ell$ be a line whose left side contains exactly $z$ points of $P_i'\cap \triangle$. For example, $\ell$ can be chosen as a vertical line between the $z$-th leftmost point and the $(z+1)$-th leftmost point of $P_i'\cap \triangle$ (if the two points are on the same vertical line, then we slightly perturb the line so that its left side contains exactly $z$ points of $P_i'\cap \triangle$). Instead of arbitrarily picking $z$ points of $P_i'\cap \triangle$ to form $P_{i+1}$, we pick
the $z$ points to the left of $\ell$.
We now use $\triangle_{i+1}$ to refer to the region of $\triangle$ to the left of $\ell$, which is either a triangle or a convex quadrilateral.
%This finishes the construction of $(P_{i+1},\triangle_{i+1})$.

Since each cell $\triangle_{i+1}$ is only a subset of its counterpart in the original algorithm, the crossing number of our partition is also $O(\sqrt{n})$. We still use $\Pi=\{(P_1,\triangle_1),\ldots,(P_m,\triangle_m)\}$ with $m=\Theta(r)$ to denote our partition.
All the properties of the enhanced simplical partition hold for $\Pi$. In particular, the following lemma proves that the weakly-overlapped property holds.

\begin{mylemma}\label{lem:property4}{\em (The weakly-overlapped property)}
For any cell $\triangle_i$ of $\Pi$, if $\triangle_i$ contains a point $p\in P_j$ with $j\neq i$, then all points of $P_i$ are outside $\triangle_j$.
\end{mylemma}
\begin{proof}
Suppose $\triangle_i$ contains a point $p\in P_j$ with $j\neq i$.
When the algorithm constructs $P_i$ in the $i$-th iteration, $\triangle_i$ does not contain any point of $P_{i-1}'\setminus P_i$. Hence, $P_j$ must be constructed earlier than $P_i$, i.e., $j<i$. When the algorithm constructs $P_j$ in the $j$-th iteration, $\triangle_j$ does not contain any point of $P_{j-1}'\setminus P_j$. Since $j<i$, $P_i\subseteq P_{j-1}'\setminus P_j$. Therefore, $\triangle_j$ does not contain any point of $P_i$.
\end{proof}

The next lemma shows that the algorithm can be implemented in $O(n^{1.5})$ time.
\begin{mylemma}\label{lem:time10}
The enhanced simplicial partition $\Pi$ can be computed in $O(n^{1.5})$ time.
\end{mylemma}
\begin{proof}
As discussed before, the first main step runs in $O(n\sqrt{r})$ time,
which is bounded by $O(n^{1.5})$ as $r=\sqrt{n}$. Below we
discuss the second main step.

The second main step has $m$ iterations. In each iteration, we need to
compute the $(1/t_i)$-cutting $\Xi_i$ for $(H,w_i)$, which can be done
in $O(r\cdot t_i)$ time by Lemma~\ref{lem:cutting} since $|H|=r$.
This is $O(r^{3/2})$ time, for $t_i=\Theta(\sqrt{n_i/z})$ and
$n_i/z\leq n/z=r$. However, we cannot apply Lemma~\ref{lem:cutting}
directly to compute $\Xi_i$ as the weights of the lines of $H$ might
be too large.
Matou\v{s}ek (in Lemma 3.4~\cite{ref:MatousekEf92}) suggested a method
that can resolve the issue when $r$ is a constant.  In
Lemma~\ref{lem:weightcutting}, we extend the method and show that
$\Xi_i$ can be computed in $O(r^{3/2})$ time.

After $\Xi_i$ is obtained, we need to find a triangle $\triangle^*$ of
$\Xi_i$ that contains at least $z$ points. One approach
is to first build a point location data structure on
$\Xi_i$~\cite{ref:EdelsbrunnerOp86,ref:KirkpatrickOp83} and then use
it to find the triangle of $\Xi_i$ that contains each point of $P_i'$.
The total time is $O(r+n_i\log r)$. However, this would lead to an
overall time of $O(nr\log r)$ for all $m$ iterations, which is not
bounded by $O(n^{1.5})$. We can improve the algorithm in the following
way. We build a simplex range reporting data structure on $P$ before
the first iteration. For example, we can use
Matou\v{s}ek's approach in~\cite{ref:MatousekRa93}, which builds a
data structure of $O(n)$ space in $O(n^{1+\delta})$ that can answer
each simplex range reporting query on $P$ in $O(\sqrt{n}+k)$ time,
where $k$ is the number of points of $P$ in the query
simplex\footnote{Because we can afford a preprocessing time of
$O(n^{1.5})$, we could use a simpler approach as long as the space is
$O(n)$ and the query time is $O(\sqrt{n}+k)$.}.
Then, for each triangle $\triangle$ of $\Xi_i$, using a
simplex range reporting query, we find all points of $P$ in
$\triangle$, and for each point we determine whether it is in $P_i'$
(for this we could put a mark on each point of $P_i'$). In this way,
we can determine the number of points $P_i'$ in $\triangle$ in
$O(\sqrt{n}+k)$ time. Doing this for all triangles of $\Xi_i$ takes
$O(r\sqrt{n}+n)$ time in total as $\Xi_i$ has at most $r$ triangles.
Subsequently, we can determine $\triangle^*$, after which we can
obtain the cell $\triangle_{i+1}$ and the subset $P_{i+1}$ in
additional $O(n)$ time. In summary, we can compute $(\triangle_{i+1},
P_{i+1})$ in $O(r\sqrt{n}+n)$ time.

Next we update the crossing numbers of the lines of $H$. For each line $l\in H$, if $l$ crosses $\triangle_{i+1}$, then $k_{i+1}(l)=k_i(l)+1$; otherwise, $k_{i+1}(l)=k_i(l)$. This steps takes $O(r)$ time.

This finishes the $i$-th iteration, which takes
$O(r^{3/2}+r\sqrt{n}+n)$ time in total. As $r=\sqrt{n}$ and there are $O(r)$ iterations, the total time of the algorithm is $O(n^{1.5})$.
\end{proof}

\begin{mylemma}\label{lem:weightcutting}
Suppose the crossing numbers $k_i(l)$'s are known for all lines $l\in
H$. Then, we can compute the cutting $(1/t_i)$-cutting $\Xi_i$ for
$(H,w_i)$ in $O(r^{3/2})$ time.
\end{mylemma}
\begin{proof}
We extend the method suggested by Matou\v{s}ek (in Lemma
3.4~\cite{ref:MatousekEf92}) and the algorithm in Theorem 2.8
of~\cite{ref:MatousekCu91} for computing a cutting for a set
of weighted lines.

Recall that $w_i(H)=\sum_{l\in H}w_i(l)=\sum_{l\in H}2^{k_i(l)}$. We first
determine an integer $a$ such that $2^a\leq w_i(H)<2^{a+1}$.
Matou\v{s}ek (in Lemma 3.2~\cite{ref:MatousekEf92}) already proved
that $\log w_i(H)=O(\sqrt{r})$. Hence, $a+1\leq c\cdot \sqrt{r}$ for a
sufficiently large constant $c$. This also implies $k_i(l)\leq c\cdot
\sqrt{r}$ for each $l\in H$. We can compute $a$ in $O(r^{3/2})$
time as follows.

Let $A$ be a array of size $c\cdot \sqrt{r}$. Initially, every element
of $A$ is $0$. Let $value(A)$ denote the value of the binary code of
the elements of $A$ (each element of $A$ is either $1$ or $0$; note that $value(A)$ is only used for discussion). So initially $value(A)=0$.
For each $l\in H$, we add $2^{k_i(l)}$ to $value(A)$ by updating the
array $A$. Since $k_i(l)\leq c\cdot \sqrt{r}$, the addition operation
can be easily done in $O(\sqrt{r})$ time by scanning the array. As $|H|=r$, the total time for
doing this for all lines of $H$ is $O(r^{3/2})$. Finally, if $i$ is the
largest index of $A$ with $A[i]=1$, then we have $a=i$.

Let $b=\lfloor\log r\rfloor$. Thus, $2^b\leq r\leq 2^{b+1}$.

We define a multiset $H'$ as follows. For each line $l\in H$, if
$b+1+k_i(l)-a\geq 0$, then we put $2^{b+1+k_i(l)-a}$ copies of $l$ in
$H'$; otherwise, we put just one copy of $l$ in $H'$. Let $|H'|$
denote the cardinality of $H'$, counted with the multiplicities. We
have the following:
\begin{equation*}
\begin{split}
|H'| & \leq   |H| + \sum_{l\in H} 2^{b+1+k_i(l)-a} = r +
2^{b+1-a} \cdot  \sum_{l\in H} 2^{k_i(l)} = r + 2^{b+1-a} \cdot w_i(H)
\\
& \leq  r + 2^{b+1-a} \cdot 2^{a+1} = r + 2^{b+2}\leq r + 4r = 5r.
\end{split}
\end{equation*}

%This also implies that $b+1+k_i(l)-a=O(\log r)$ for each $l\in H$.
%Hence, $O(\log r)$ time is more than sufficient to perform
This also implies that the step of ``put $2^{b+1+k_i(l)-a}$ copies of $l$ in $H'$'' for all $l\in H$ can be done in $O(r)$ time. Therefore, generating the multiset $H'$ takes $O(r)$ time.

Now we compute a $\frac{1}{5\sqrt{r}}$-cutting $\Xi$ for the
unweighted multiset $H'$ in $O(r^{3/2})$ time by Lemma~\ref{lem:cutting}. In what follows, we prove
that $\Xi$ is a $(1/t_i)$-cutting for the weighted set $(H,w)$. Thus,
we can simply return $\Xi$ as $\Xi_i$. The total time of the algorithm
is $O(r^{3/2})$. This will prove the lemma.

As $t_i=\Theta(\sqrt{r})$, our goal is to show that
that $\Xi$ is a $\frac{1}{\sqrt{r}}$-cutting for $(H,w)$.
Let $\triangle$ be a triangle of $\Xi$. Define $H_{\triangle}$ to
be the subset of lines of $H$ that cross $\triangle$. It is
sufficient to prove $w_i(H_{\triangle})\leq w_i(H)/\sqrt{r}$.

Let $H'_{\triangle}$ denote the multiset of lines of $H'$ crossing
$\triangle$.
%counting multiplicties.
Because $\Xi$ is a $\frac{1}{5\sqrt{r}}$-cutting of $H'$ and $|H'|\leq 5r$, it holds
that $|H'_{\triangle}|\leq \frac{|H'|}{5\sqrt{r}}\leq \sqrt{r}$.
Consequently, we can derive:
\begin{equation*}
\begin{split}
w_i(H_{\triangle}) & = \sum_{l\in H_{\triangle}} w_i(l) = \sum_{l\in
H_{\triangle}} 2^{k_i(l)} = \frac{1}{2^{b+1-a}} \cdot \sum_{l\in
H_{\triangle}} 2^{b+1+k_i(l)-a} \leq  \frac{1}{2^{b+1-a}} \cdot
|H'_{\triangle}| \\
& \leq \frac{\sqrt{r}}{2^{b+1-a}} = \frac{2^a\cdot \sqrt{r}}{2^{b+1}}
\leq \frac{w_i(H)\cdot \sqrt{r}}{2^{b+1}} \leq \frac{w_i(H)\cdot
\sqrt{r}}{r} = \frac{w_i(H)}{\sqrt{r}}.\\
\end{split}
\end{equation*}

This proves that $\Xi$ is a $\frac{1}{\sqrt{r}}$-cutting for $(H,w)$.
\end{proof}

In the following, we will preprocess each subset $P_i$ of $\Pi$ by using/modifying the basic algorithm in~\cite{ref:MatousekRa93}. But before that, we give a picture on how we will use our simplicial partition to store edges of $S$ to solve our segment detection and ray-shooting queries.

\subparagraph{Storing the segments in $\boldsymbol{\Pi}$.}
For each segment $s$ of $S$, if both endpoints of $s$ are in the same subset $P_i$ of $\Pi$, then $s$ is in the cell $\triangle_i$ as $\triangle_i$ is convex and we store $s$ in $\triangle_i$; let $S_i$ denote the set of segments stored in $\triangle_i$. Otherwise, let $P_i$ and $P_j$ be the two subsets that contain the endpoints of $s$, respectively. The weakly-overlapped property in Lemma~\ref{lem:property4} leads to the following observation.

\begin{observation}\label{obser:boundary}
The segment $s$ intersects the boundary of at least one cell of $\triangle_i$ and $\triangle_j$.
\end{observation}
\begin{proof}
If $s$ intersects the boundary of $\triangle_i$, then the observation follows. Otherwise, both endpoints of $s$ are in $\triangle_i$. Let $p$ be the endpoint of $s$ that is in $P_j$ and let $q$ be the other endpoint, which is in $P_i$. Since $\triangle_i$ contains $p$, by Lemma~\ref{lem:property4}, all points of $P_i$ are outside $\triangle_j$. Hence, $q$ is outside $\triangle_j$, implying that $s$ must intersect the boundary of $\triangle_j$.
\end{proof}

By Observation~\ref{obser:boundary}, we find a cell $\triangle$ of $\triangle_i$ and $\triangle_j$ whose boundary intersects $s$. Let $e$ be an edge of $\triangle$ that intersects $s$. We store $s$ at $e$; let $S(e)$ denote the set of segments of $S$ that are stored at $e$. In this way, each segment of $S$ is stored exactly once. Next, for each cell $\triangle\in \Pi$ and for each edge $e$ of $\triangle$, we preprocess $S(e)$ using Lemma~\ref{lem:anchor}(1) or using Lemma~\ref{lem:anchor}(2) if the segments of $S$ are nonintersecting. With $\Pi$, the above preprocessing on $S$ takes $O(n\log n)$ time and $O(n)$ space.
Later in Section~\ref{sec:subsetprocess} we will prove the following lemma.

\begin{mylemma}\label{lem:subsetprocess}
\begin{enumerate}
\item
For each subset $P_i$ of $\Pi$, with $O(|P_i|^{2})$ time and $O(|P_i|)$ space preprocessing, we can determine whether a query line intersects any segment of $S_i$ in $O(\sqrt{|P_i|}\log |P_i|)$ time.
\item
If the segments of $S_i$ are nonintersecting, then with $O(|P_i|^{2})$ time and $O(|P_i|)$ space preprocessing, we can determine the first segment of $S_i$ hit by a query ray in $O(\sqrt{|P_i|}\log |P_i|)$ time.
\end{enumerate}
\end{mylemma}

We can thus obtain our results for the segment intersection problem and the ray-shooting problem.

\begin{theorem}\label{theo:segmentray}
\begin{enumerate}
\item
Given a set of $n$ (possibly intersecting) line segments, we can build a data structure of space $O(n)$ in $O(n^{1.5})$ time so that whether a query line intersects any segment can be determined in $O(\sqrt{n}\log n)$ time.
\item
Given a set of $n$ (possibly intersecting) line segments, we can build a data structure of space $O(n\log n)$ in $O(n^{1.5})$ time so that the first segment hit by a query ray can be found in $O(\sqrt{n}\log n)$ time.
\item
Given a set of $n$ nonintersecting line segments, we can build a data structure of space $O(n)$ in $O(n^{1.5})$ time so that the first segment hit by a query ray can be found in $O(\sqrt{n}\log n)$ time.
\end{enumerate}
\end{theorem}
\begin{proof}
We begin with Part (1) of the theorem.
For the preprocessing time, computing $\Pi$ takes $O(n^{1.5})$ time. Storing the segments in $\Pi$ and preprocessing them by Lemma~\ref{lem:anchor} takes $O(n\log n)$ time. Applying Lemma~\ref{lem:subsetprocess} on all subsets $P_i$ of $\Pi$ takes $O(n^{1.5})$ time in total, as the size of each $P_i$ is $O(\sqrt{n})$. Hence, the overall preprocessing time is $O(n^{1.5})$. Following the same analysis, the space is $O(n)$.
%We have discussed the preprocessing time and space.
Next we describe the query algorithm and analyze the query time.

Consider a query line $\ell$. First, for each cell $\triangle_i$ of $\Pi$, for each edge $e$ of $\triangle_i$, we determine whether $\ell$ intersects a segment of $S(e)$, which can be done in $O(\log n)$ time by Lemma~\ref{lem:anchor}(1); if the answer is yes, then we halt the entire query algorithm. As $\Pi$ has $\Theta(\sqrt{n})$ cells and each cell has at most four edges, the total time of this step is $O(\sqrt{n}\log n)$. Second, by checking every cell of $\Pi$, we find those cells that are crossed by $\ell$. For each such cell $\triangle_i$, by Lemma~\ref{lem:subsetprocess}(1), we determine whether $\ell$ intersects any segment of $S_i$ in $O(n^{1/4}\log n)$ time, for $|P_i|=\Theta(\sqrt{n})$; if the answer is yes, then we halt the entire algorithm. As $\ell$ can cross at most $O(n^{1/4})$ cells of $\Pi$, this step takes $O(\sqrt{n}\log n)$ time. Hence, the query time is $O(\sqrt{n}\log n)$.

To see the correctness of the algorithm, suppose $\ell$ intersects a segment $s\in S$. If both endpoints of $s$ are in the same subset $P_i$ of $\Pi$, then $s\in S_i$ and $\ell$ must cross the cell $\triangle_i$ and thus the intersection will be detected in the second step of the algorithm when we invoke the query algorithm of Lemma~\ref{lem:subsetprocess}(1) on $P_i$. If the two endpoints of $s$ are not in the same subset $P_i$ of $\Pi$, then by Observation~\ref{obser:boundary}, $s$ must be stored at an edge $e$ of a cell of $\Pi$; thus the intersection will be detected when we invoke the query algorithm of Lemma~\ref{lem:anchor}(1) on $S(e)$.

Part (2) of the theorem has been discussed in the proof of Theorem~\ref{theo:randomray} (see the remark at the end of the proof), i.e., we apply Cheng and Janardan's algorithmic scheme~\cite{ref:ChengAl92} but instead use our result in Theorems~\ref{theo:raylines}  for the ray-shooting problem among lines and use the result of Part (1) of this theorem for the intersection detection problem.

For Part (3), the preprocessing is similar to Part (1).
The query algorithm is also very similar.
Consider a query ray $\rho$. First, for each cell $\triangle_i$ of $\Pi$, for each edge $e$ of $\triangle_i$, we determine the first segment of $S(e)$ hit by $\rho$, which can be done in $O(\log n)$ time by Lemma~\ref{lem:anchor}(2). Second, for each cell $\triangle_i$ of $\Pi$, if it is crossed by $\Pi$, then by Lemma~\ref{lem:subsetprocess}(2), we find the first segment of $S_i$ hit by $\rho$ in $O(n^{1/4}\log n)$ time. Third, among all segments found above, we return the one whose intersection with $\rho$ is closest to the origin of $\rho$. The total query time is $O(\sqrt{n}\log n)$ time.
\end{proof}

%We could preprocess $P_i$ recursively using the above algorithm until each subset has a constant number of points.

\subsubsection{Proving Lemma~\ref{lem:subsetprocess}}
\label{sec:subsetprocess}

In this section, we prove Lemma~\ref{lem:subsetprocess}. Since both endpoints of $s$ are in $P_i$ for
each segment $s\in S_i$, $|S_i|\leq |P_i|/2$. To simplify the notation, let $n=|P_i|$, $P=P_i$,
and $S=S_i$. Hence, $|S|\leq n/2$. With these notation, we restate Lemma~\ref{lem:subsetprocess} as follows.

\begin{mylemma}\label{lem:subsetprocess10}{\em (A restatement of Lemma~\ref{lem:subsetprocess})}
Let $P$ be a set of $n$ points in the plane and let $S$ be a set of segments whose endpoints are in $P$.
\begin{enumerate}
\item
With $O(n^{2})$ time and $O(n)$ space preprocessing, whether a query line intersects any segment of $S$ can be determined in $O(\sqrt{n}\log n)$ time.
\item
If the segments of $S$ are nonintersecting, then with $O(n^{2})$ time and $O(n)$ space preprocessing, the first segment of $S$ hit by a query ray can be found in $O(\sqrt{n}\log n)$ time.
\end{enumerate}
\end{mylemma}

In the following, we prove Lemma~\ref{lem:subsetprocess10}.
We resort to the techniques of Matou\v{s}ek~\cite{ref:MatousekRa93},
which provides a more efficient partition tree using Chazelle's
algorithm for computing hierarchical cuttings~\cite{ref:ChazelleCu93}.
We still need to modify the algorithm
in~\cite{ref:MatousekRa93} as we did before for computing the enhanced
simplicial partition. In particular, we need to have a similar weakly-overlapped property.
%as that in Lemma~\ref{lem:property4}.
We also have to
change the weight function defined on the line sets in order to
achieve the claimed query time. In the following, we first review the
algorithm of Matou\v{s}ek in~\cite{ref:MatousekRa93}. As discussed
before, Matou\v{s}ek first gave a basic algorithm of polynomial time
and then reduce the time to $O(n^{1+\delta})$ using multilevel data
structures. Here we cannot use multilevel data structures and thus only use
his basic algorithm (i.e., the one in Theorem 4.1 of~\cite{ref:MatousekRa93}).
We will show that his basic algorithm can be
implemented in $O(n^{2})$ time.

We first construct a data structure for a subset $P'$ of at least half
points of $P$. To build a data structure for the whole $P$, the above
construction is performed for $P$, then for $P\setminus P'$, etc., and
thus a logarithmic number of data structures with geometrically
decreasing sizes will be obtained. Because the preprocessing time of
the data structure for $P'$ is $\Omega(n)$ and the space is $\Theta(n)$,
constructing all data structures for $P$ takes asymptotically the same time and space
as those for $P'$ only. To answer a simplex range query on $P$, each of these data
structures will be called. Since the query time for $P'$ is
$\Omega(\sqrt{n})$, the total query time for $P$ is asymptotically the
same as that for $P'$. Below we describe the data
structure for $P'$.

The data structure has a set of (not necessarily disjoint) triangles,
$\Psi_0=\{\triangle_1,\ldots,\triangle_t\}$ with $t=\sqrt{n}\log n$.
For each $1\leq i\leq t$, we have a subset $P_i\subseteq P$ of at most
$\frac{n}{2t}$ points that are contained in $\triangle_i$. The subsets $P_i$'s form a disjoint partition of
$P'$.
For each $i$, there is a rooted tree $T_i$ whose nodes correspond to
triangles, with $\triangle_i$ as the root. Each internal node of $T_i$
has $O(1)$ children whose triangles are
interior-disjoint and together cover their parent triangle. For each
triangle $\triangle$ of $T_i$, let $P(\triangle)=P_i\cap \triangle_i$.
If $\triangle$ is a leaf, then the points of $P(\triangle)$ are
explicitly stored at $\triangle$. Each point
of $P_i$ is stored in exactly one leaf triangle of $T_i$.
The depth of $T_i$ is $q=O(\log n)$. Hence, the data structure is a forest of $t$
trees. Let $\Psi_j$ denote the set of all triangles of all trees
$T_i$'s that lie at distance $j$ from the root (note that $\Psi_0$ is
consistent with this definition). For any line $l$ in the plane, let $K_j(l)$ be
the set of triangles of $\Psi_j$ crossed by $l$; let $L_j(h)$ be the
leaf triangles of $K_j(l)$. Define $K(l)=\bigcup_{j=0}^q K_j(l)$ and
$L(l)=\bigcup_{j=0}^q L_j(l)$. Matou\v{s}ek~\cite{ref:MatousekRa93}
proved that $\sum_{j=0}^q|\Psi_j|=O(n)$, and $|K(l)|=O(\sqrt{n})$ and $\sum_{\triangle\in L(l)}|P(\triangle)|=O(\sqrt{n})$ hold for any line $l$ in the plane.

%\begin{equation}\label{equ:kl}
%|K(l)|=O(\sqrt{n}), \ \ \ \sum_{\triangle\in L(l)}|P(\triangle)|=O(\sqrt{n}).
%\end{equation}

We next review Matou\v{s}ek's basic
algorithm~\cite{ref:MatousekRa93} for constructing the data structure described above.
As in the algorithm for constructing simplicial partitions, the first step is to compute a test set $H$ (called a {\em guarding set} in~\cite{ref:MatousekRa93}) of $n$ lines, which can be done in $O(n\sqrt{n})$ time as discussed in Section~\ref{sec:simpartition}. After that, the algorithm proceeds in $t$ iterations; in the $i$-th iteration, $T_i$, $\triangle_i$, and $P_i$ will be produced.

Suppose $T_j$, $\triangle_j$, and $P_j$ for all $j=1,2\ldots,i$ have been constructed.
Define $P_i'=P\setminus(P_1\cup\cdots \cup P_i)$. If $|P_i'|<n/2$, then we stop the construction. Otherwise, we proceed with the $(i+1)$-th iteration as follows. Let $\Psi_0^{(i)},\ldots,\Psi_q^{(i)}$ denote the already constructed parts of $\Psi_0,\ldots,\Psi_q$. Define $K_j^{(i)}(l)$ and $L_j^{(i)}(l)$ similarly as $K_j(l)$ and $L_j(l)$. We define a weighted line set $(H,w_i)$.
For each line $l\in H$, define a weight
\begin{equation}\label{equ:weight}
w_i(l)= \exp\bigg(\frac{\log n}{\sqrt{n}}\cdot\bigg[\sum_{j=0}^{q}4^{q-j}\cdot |K_j^{(i)}(l)|+\sum_{\triangle\in K_q^{(i)}(l)}{|P(\triangle)|}\bigg]\bigg).
\end{equation}

The next step is to compute an {\em efficient hierarchical
$(1/r)$-cutting} for $(H,w_i)$ with $r=\sqrt{n}$, which consists of a
sequence of cuttings $\Xi_0,\Xi_1,\ldots,\Xi_k$ that satisfy the
following properties. (1) $\Xi_0$ is a single triangle that contains
the entire plane. (2) For two fixed constants $C$ and $\rho>4$, for each
$1\leq j\leq k$, $\Xi_j$ is a $(1/\rho^j)$-cutting for $(H,w_i)$ of size
$O(\rho^{2j})$ such that each triangle of $\Xi_j$ is contained in a
triangle of $\Xi_{j-1}$ and each triangle of $\Xi_{j-1}$ contains at
most $C$ triangles of $\Xi_j$ (if a triangle $\triangle\in \Xi_{j-1}$ contains a triangle $\triangle'\in \Xi_j$, we say that $\triangle$ is the {\em parent} of $\triangle'$ and $\triangle'$ is a {\em child} of $\triangle$). (3) $\rho^{k-1} < r \leq \rho^k$ and
thus $k=\Theta(\log r)$.

We let $p$ be the largest index such that the size of $\Xi_p$ is at most $t$. As the size of $\Xi_j$ is
$O(\rho^{2j})$, we obtain that $\rho^{2p}=\Theta(t)$ and $\Xi_p$ is a
$(1/r_p)$-cutting of $(H,w_i)$ with $r_p=\rho^p=\Theta(\sqrt{t})$.
We define $q=k-p$. Note that $\rho^{q}=O(r/\sqrt{t})=O(\sqrt{n/t})$. Since $|P_i'|\geq n/2$ and $\Xi_p$ has at most $t$ triangles, $\Xi_p$ has a triangle, denoted by $\triangle_{i+1}$, containing at least $\frac{n}{2t}$ points of $P_i'$. We arbitrarily select $\frac{n}{2t}$ points of $P_i'\cap \triangle_{i+1}$ to form the set $P_{i+1}$. Further, all triangles in $\Xi_p,\Xi_{p+1},\ldots,\Xi_k$ contained in $\triangle_{i+1}$ form the tree $T_{i+1}$, whose root is $\triangle_{i+1}$. Next, we remove some nodes from $T_{i+1}$ as follows; we call it a {\em pruning procedure}. Starting from the root, we perform a depth-first-search (DFS). Let $\triangle$ be the triangle of the current node the DFS is visiting. Suppose $\triangle$ belongs to $\Xi_{a+j}$ for some $0\leq j\leq q$. If $\triangle$ contains at least $2^{q-j}$ points of $P_{i+1}$ ($\triangle$ is called a {\em fat} triangle in~\cite{ref:MatousekRa93}), then we proceed on the children of $\triangle$; otherwise, we make $\triangle$ a leaf node and return to its parent (and continue DFS). In other words, a triangle of $T_{i+1}$ is kept if and only all its ancestor triangles are fat. This finishes the construction of the $(i+1)$-th iteration.

For our purpose, we modify the algorithm as follows (we only point out the differences). Let $\triangle^*$ denote the above $\triangle_{i+1}$ that contains at least $\frac{n}{2t}$ points of $P_i'$. Let $l^*$ be a line such that its left side contains exactly $\frac{n}{2t}$ points of $P_i'\cap \triangle^*$ (and we use these points to form $P_{i+1}$).
%(we can find the line in the same way as we did before for constructing the enhanced simplicial partition).
We now set $\triangle_{i+1}$ to the part of $\triangle^*$ on the left side of $l^*$. Hence, $\triangle_{i+1}$ is either a triangle or a convex quadrilateral. We form the tree $T_{i+1}$ in the same way as above except that each node of $T_{i+1}$ now corresponds to a {\em cell}, which is either a triangle or a convex quadrilateral. This change will guarantee a similar weakly-overlapped property as in Lemma~\ref{lem:property4}.

The second change we make is that we set $t$ to $\sqrt{n}$ instead of $\sqrt{n}\log n$. The third change is that we redefine the weight function in \eqref{equ:weight} as follows (i.e., the second term does not have the $\log n$ factor any more):
\begin{equation}\label{equ:newweight}
w_i(l)= \exp\bigg(\frac{\log n}{\sqrt{n}}\cdot\sum_{j=0}^{q}4^{q-j}\cdot |K_j^{(i)}(l)|+\frac{1}{\sqrt{n}}\cdot\sum_{\triangle\in K_q^{(i)}(l)}{|P(\triangle)|}\bigg).
\end{equation}
As a consequence, by following Matou\v{s}ek's proof in~\cite{ref:MatousekRa93} (i.e., Theorem 4.1), we have the following Lemma~\ref{lem:bound}. Before proceeding to the lemma proof, we briefly explain why we need to make these changes. As will be clear later, the time complexity of the query algorithm for our problem is bounded by $O(t\log n + K(l)\cdot \log n + \sum_{\triangle\in L(l)}|P(\triangle)|)$. To guarantee the $O(\sqrt{n}\log n)$ query time, we need to make sure that both $t$ and $K(l)$ are bounded by $O(\sqrt{n})$. For the simplex range searching problem, Matou\v{s}ek's algorithm needs to bound both $K(l)$ and $\sum_{\triangle\in L(l)}|P(\triangle)|$ by $O(\sqrt{n})$, and to do so, the algorithm needs to set $t$ to $\sqrt{n}\log n$. For our problem, it is sufficient to bound $\sum_{\triangle\in L(l)}|P(\triangle)|$ by $O(\sqrt{n}\log n)$\footnote{This is also reflected in our new weight function, where the second term does not have a $\log n$ factor as in \eqref{equ:weight}; intuitively, this implies that the number of points in the leaves is less important than before.}; consequently, we are able to use a smaller $t$ with $t=\sqrt{n}$.
%\footnote{We briefly explain the changes of the proof by referring to~\cite{ref:MatousekRa93}. In Lemma~4.2, the right-hand side of each equality now has an additional factor $\sqrt{\log n}$. To prove these new equalities, the last item of the equation at the bottom of Page 168 now becomes $O(n^{1-1/d}\sqrt{\log n})$, and the two $\log n$ factors in the equation now both become $\sqrt{\log n}$. The rest of the proof for proving $\log w_t(H)=O(\log n)$ keeps exactly the same as before but use the new weights and the new value of $t$ instead (i.e., the $\log n$ factors in all equations in Page 169 becomes $\sqrt{\log n}$).}),

\begin{mylemma}\label{lem:bound}
\begin{enumerate}
\item $\sum_{j=0}^q|\Psi_j|=O(n)$.
\item
For any line $l$ in the plane, $|K(l)|=O(\sqrt{n})$ and $\sum_{\triangle\in L(l)}|P(\triangle)|=O(\sqrt{n}\log n)$.
\end{enumerate}
%\begin{equation}\label{equ:kl}
%|K(l)|=O(\sqrt{n}), \ \ \ \sum_{\triangle\in L(l)}|P(\triangle)|=O(\sqrt{n}\log n).
%\end{equation}
\end{mylemma}
\begin{proof}
%Also, by exactly the same proof as in~\cite{ref:MatousekRa93}, $\sum_{j=0}^q|\Psi_j|=O(n)$ still holds, implying the space is still $O(n)$.
The proof is almost the same as that in~\cite{ref:MatousekRa93} (i.e., the proof of Theorem 4.1). We briefly discuss it by referring to the corresponding parts in~\cite{ref:MatousekRa93} .

The proof for $\sum_{j=0}^q|\Psi_j|=O(n)$ is exactly the same as that in~\cite{ref:MatousekRa93}. Indeed, the algorithm adds $O(\rho^{2q})=O(n/t)$ new cells in each of the $t$ iterations. Therefore, the total number of cells is $O(n)$.

For the second lemma statement, we claim that for any line $l\in H$ the following hold (which correspond to Lemma~4.2~\cite{ref:MatousekRa93}):
\begin{equation}\label{equ:10}
|K_j(l)|=O(\sqrt{n}\cdot 4^{-(q-j)}), j=0,1,\ldots,q,
\end{equation}
\begin{equation}\label{equ:20}
\sum_{\triangle\in K_q(l)}|P(\triangle)|=O(\sqrt{n}\log n).
\end{equation}

With the above claim, following literally the same proof as that in~\cite{ref:MatousekRa93} (specifically, the three paragraphs after Lemma 4.2~\cite{ref:MatousekRa93}), the second lemma statement can be proved.

In the following, we prove the above claim, which is similar to the proof of Lemma 4.2 of~\cite{ref:MatousekRa93}. We focus on the differences.

The key is to prove that $\log w_t(H)=O(\log n)$ (recall that $w_t(H)$
stands for the total weight of all lines of $H$ after the $t$-th
iteration of the algorithm). Indeed, by our definition of the weight
funciton, we have
\begin{equation*}
\frac{\log n}{\sqrt{n}}\cdot\sum_{j=0}^{q}4^{q-j}\cdot
|K_j(l)|+\frac{1}{\sqrt{n}}\cdot\sum_{\triangle\in
K_q(l)}{|P(\triangle)|}\leq \log w_t(H),\ \  j=0,1,\ldots,q.
\end{equation*}
This leads to Equations~\eqref{equ:10} and \eqref{equ:20}, for  $\log w_t(H)=O(\log n)$.

It remains to prove $\log w_t(H)=O(\log n)$. The proof follows the same line as
in~\cite{ref:MatousekRa93}. Indeed, the bound for $f_j$ (see~\cite{ref:MatousekRa93} for the definition)
is the same as
before as it is for the first term of \eqref{equ:newweight}, which is
the same as Matou\v{s}ek's weight definition in \eqref{equ:weight}. The bound for $f(\triangle)$ (which is $f(s)$ in~\cite{ref:MatousekRa93}), however,
is different because our weight definition does not have the $\log n$ factor.
As a consequence, we have the following
$$f(\triangle)=1+O\bigg(\frac{\exp(|P(\triangle)|/\sqrt{n})-1}{\sqrt{n}}\bigg)$$
Note that $|P(\triangle)|\leq n/(2t)=\sqrt{n}/2$. Using the
inequalities $1+x\leq e^x\leq 1+2x$ (the latter one holds for $x\leq
1$ \footnote{To guarantee $|P(\triangle)|/\sqrt{n}\leq 1$ for using the inequality $e^x\leq 1+2x$, it suffices to have $n/(2t)\leq \sqrt{n}$. Hence, $t\geq \sqrt{n}/2$. Therefore, $\sqrt{n}/2$ is the smallest possible value for $t$ to make the proof work if we choose the weight function as \eqref{equ:newweight}. Using Matou\v{s}ek's original weight function, the smallest possible value for $t$ is $\sqrt{n}\log n/2$. Therefore, in order to set $t$ to $\sqrt{n}$ (to guarantee the query time complexities of our problems), we have to change the weight function in order to make sure the same proof works.\label{foot:10}}), we further obtain
$$f(\triangle)=1+O\bigg(\frac{\exp(|P(\triangle)|/\sqrt{n})-1}{\sqrt{n}}\bigg)\leq
1+O\bigg(\frac{|P(\triangle)|/\sqrt{n}}{\sqrt{n}}\bigg)\leq \exp\bigg(O\bigg(\frac{|P(\triangle)|}{n}\bigg)\bigg).$$

Following the rest of the argument in~\cite{ref:MatousekRa93}, we can still derive $\log w_t(H)=O(\log n)$\footnote{Note that Matou\v{s}ek~\cite{ref:MatousekRa93} also showed that the weight of each line of $H$ increases by at most a constant factor in every iteration. This property does not hold any more in our case. However, this does not affect the proof of $\log w_t(H)=O(\log n)$, i.e., although we do not have a good bound for the increase of the weight in each individual iteration, we can still achieve asymptotically the same bound as before for the total weight after all iterations.}.
\end{proof}

This finishes our algorithm for constructing the data structure for $P'$. As discussed before, to construct the data structure for the whole set $P$, we perform the above construction for a logarithmic number of times; each time we obtain a forest. The total number of all trees in all these forests is at most a number $f\leq 2t$. We order these trees by the time they constructed: $T_1,T_2,\ldots,T_f$. Correspondingly, we have the cells $\triangle_1,\ldots,\triangle_f$, and the subsets $P_1,\ldots,P_f$, which form a disjoint partition of $P$. Because the sizes of the problems which these logarithmic number of constructions are based on are geometrically decreasing, the bounds in Lemma~\ref{lem:bound} still hold for all these $f$ trees.
The following lemma is analogous to Lemma~\ref{lem:property4}.

\begin{mylemma}\label{lem:property4new}{\em (The weakly-overlapped property)}
Among the cells $\triangle_1,\ldots,\triangle_f$, if a cell $\triangle_i$ contains a point $p\in P_j$ with $j\neq i$, then all points of $P_i$ are outside $\triangle_j$.
\end{mylemma}
\begin{proof}
The proof is literally the same as that for Lemma~\ref{lem:property4}. Suppose $\triangle_i$ contains a point $p\in P_j$ with $j\neq i$. When the algorithm constructs $P_i$, $\triangle_i$ does not contain any point of $P_{i-1}'\setminus P_i$, where $P_{i-1}'=P\setminus(P_1\cup\cdots \cup P_{i-1})$. Hence, $P_j$ must be constructed earlier than $P_i$, i.e., $j<i$. When the algorithm constructs $P_j$, $\triangle_j$ does not contain any point of $P_{j-1}'\setminus P_j$, where $P_{j-1}'=P\setminus(P_1\cup\cdots \cup P_{j-1})$. Since $j<i$, $P_i\subseteq P_{j-1}'\setminus P_j$. Therefore, $\triangle_j$ does not contain any point of $P_i$.
\end{proof}

\begin{mylemma}\label{lem:preprowholeP}
The data structure for the whole $P$ can be constructed in $O(n^2)$ time and $O(n)$ space.
\end{mylemma}
\begin{proof}
As discussed before, it is sufficient to show that the data structure for $P'$ can be constructed in $O(n^2)$ time and $O(n)$ space. The $O(n)$ space follows from Lemma~\ref{lem:bound}(1). Below we bound the construction time.

As discussed before, computing the test set $H$ takes $O(n\sqrt{n})$ time. The algorithm proceeds in $t=\sqrt{n}$ iterations. Consider the $(i+1)$-th iteration.

For each line $l\in H$, define $k_i(l)$ as the exponential of its weight $w_i(l)$, i.e., $k_i(l)=\frac{\log n}{\sqrt{n}}\cdot\sum_{j=0}^{q}4^{q-j}\cdot |K_j^{(i)}(l)|+\frac{1}{\sqrt{n}}\cdot\sum_{\triangle\in K_q^{(i)}(l)}{|P(\triangle)|}$.
Note that Lemma~\ref{lem:bound} proves that $k_i(l)$ is bounded by $O(\log n)$.
%oved that $k_i(l)=O(\log n)$ (which holds for our new weight definition and the new value of $t$).
Lemma~\ref{lem:hcutting} shows that the efficient hierarchical $(1/\sqrt{n})$-cuttings for $(H,w_i)$ can be constructed in $O(n\sqrt{n})$ time in a similar way as Lemma~\ref{lem:weightcutting}.

To find the triangle $\triangle^*$ of $\Xi_p$ that contains at least $\frac{n}{2t}$ points of $P_i'$,
we first build a point location data structure on $\Xi_p$ in $O(t)$ time~\cite{ref:EdelsbrunnerOp86,ref:KirkpatrickOp83}, for $\Xi_p$ has at most $t$ triangles, and then perform a point location for each point of $P_i'$. In this way, determining $\triangle^*$ can be done in $O(t+n\log t)$ time. After that, obtaining $\triangle_{i+1}$ and the subset $P_{i+1}$ can be easily done in additional $O(n)$ time.

Next, we perform the pruning procedure by running DFS on $T_{i+1}$, which is initially formed by all cells of $\Xi_p,\ldots, \Xi_k$ contained in $\triangle_{i+1}$. To this end, we need to know the number of points of $P_{i+1}$ contained in each cell $\triangle$ of $T_{i+1}$. For this, we again apply the above point location algorithm on each $\Xi_{j}$ for $j=p,p+1,\cdots, k$. Notice that the total number of cells of all cuttings $\Xi_p,\ldots, \Xi_k$ contained in $\triangle_{i+1}$ is $\rho^{2q}=O(n/t)$, where $q=k-p$. Hence, the total time for building all point location data structures is $O(n/t)$. The total time for point location queries is $O(|P_{i+1}|\cdot \log n\cdot q)$, which is $O(\frac{n}{t}\log^2 n)$, for $|P_{i+1}|=\frac{n}{2t}$ and $q=O(\log n)$. Therefore, computing the numbers of points of $P_{i+1}$ contained in the cells of $T_{i+1}$ can be done in $O(\frac{n}{t}\log^2 n)$ time. Subsequently, running DFS on $T_{i+1}$ takes $O(|T_{i+1}|)$ time, which is $O(n/t)$ since the total number of cells of the cuttings $\Xi_p,\ldots, \Xi_k$ contained in $\triangle_{i+1}$ is $O(n/t)$.

Finally, we update the values $k_i(l)$'s for all lines $l\in H$. For each line $l\in H$, by traversing $T_{i+1}$, for each cell $\triangle$ of the tree, if $l$ crosses $\triangle$, then we can update $k_i(l)$ as follows. Suppose $l$ crosses $\triangle$ and the depth of $\triangle$ is $j$. Then, the term $|K_{j}^{(i)}(l)|$ in the weight function increases by one, and thus we simply increment $k_i(l)$ by $4^{q-j}\cdot \sqrt{\log n/n}$. If $j=q$, then $\triangle$ is a leaf and we further increase $k_i(l)$ by $|P(\triangle)|\cdot \sqrt{1/n}$; note that the size $|P(\triangle)|$ is stored at $\triangle$.
Since $|T_{i+1}|=O(n/t)$ and $|H|=n$, updating the values $k_i(l)$'s for all lines $l\in H$ can be easily done in $O(n^2/t)$ time, which is $O(n\sqrt{n})$ time.

This finishes the algorithm for the $(i+1)$-th iteration, which takes $O(n\sqrt{n})$ time.
 %dominated by  computing the efficient hierarchical $(1/\sqrt{n})$-cuttings for $(H,w_i)$.
 As there are $t=\sqrt{n}$ iterations, the total time of the algorithm is $O(n^2)$.
\end{proof}

\begin{mylemma}\label{lem:hcutting}
Suppose the values $k_i(l)$'s are known for all lines $l\in
H$. Then, we can compute an efficient hierarchical $(1/\sqrt{n})$-cutting for $(H,w_i)$ in $O(n\sqrt{n})$ time.
\end{mylemma}
\begin{proof}
The proof is very similar to that for Lemma~\ref{lem:weightcutting},
so we only point out the differences. The algorithm first compute an
integer $a$ so that $e^a\leq w_i(H)<e^{a+1}$. For a similar task, an array $A$ of
size $O(\sqrt{r})$ is used in Lemma~\ref{lem:weightcutting}. Here,
since $\log w_i(H)=O(\log n)$ by Lemma~\ref{lem:bound},
we can use an array of size $O(\log n)$. Also, $value(A)$ is defined on the elements of $A$ with base $2$ in Lemma~\ref{lem:weightcutting}; here we use base $e$.
Following the same algorithm,
we can compute $a$ in $O(n\log n)$ time. After having $a$, the rest of
the algorithm is very similar as before (e.g., we use base $e$ instead of base $2$).
Also the algorithm for Lemma~\ref{lem:weightcutting} only needs a cutting while here we need
an efficient hierarchical cutting, but they are computed by
exactly the same algorithm of Lemma~\ref{lem:cutting}. The analysis is also similar.
The total time is $O(n\sqrt{n})$ (i.e., replace $r$ in
Lemma~\ref{lem:weightcutting} by $\sqrt{n}$).
\end{proof}

%Recall that the data structure we have been working on is on the subset $P'$. For the whole $P$

In summary, we have computed $f$ trees, $T_1,\ldots, T_f$,
along with cells $\triangle_1,\cdots,\triangle_f$ and subsets
$P_1,\ldots,P_f$, with the following properties: (1) The subsets $P_i$'s are disjoint and $P=\bigcup_{i=1}^fP_i$. (2) Ech cell is either a triangle or a convex quadrilateral. (3)
Each subset $P_i$ is contained in $\triangle_i$. (4)
The weakly-overlapped property in Lemma~\ref{lem:property4new} holds. (5) The bounds of Lemma~\ref{lem:bound} hold for all $f$ trees.
We use $\Psi$ to refer to this data structure.

\subparagraph{Storing the segments in the data structure $\boldsymbol{\Psi}$.}
We now store the segments of $S$ in $\Psi$. For each segment $s\in S$,
if their endpoints are in two different subsets $P_i$ and $P_j$,
then we can prove Observation~\ref{obser:boundary} again using
Lemma~\ref{lem:property4new}. Let $\triangle$ be a cell of
$\triangle_i$ and $\triangle_j$ whose boundary intersects $s$. Let $e$
be an edge of $\triangle$ that intersects $s$. We store $s$ at $e$;
let $S(e)$ be the set of all segments stored at $e$.
If the endpoints of $s$ are in the same subset $P_i$, then we store
$e$ in the tree $T_i$ in the same way as we store segments in Chan's
partition tree in Section~\ref{sec:randomized} (indeed $T_i$ and
Chan's partition tree share similar properties: each internal node has
$O(1)$ children; children cells do not overlap and together form a
partition of their parent cell). After that, each edge $e$ of each cell of $T_i$ stores a set $S(e)$ of segments that intersect $e$.
In addition, if both endpoints of $s$ are in a leaf cell $\triangle$ of $T_i$, then we store $s$ there; let $S(\triangle)$ be the set of all segments stored in $\triangle$.
In this way, each segment is stored $O(1)$ times.

For each edge $e$ of each cell of each tree of $\Psi$, we preprocess $S(e)$ using Lemma~\ref{lem:anchor}(1), or using Lemma~\ref{lem:anchor}(2) if the segments of $S$ are nonintersecting.
After $\Psi$ is obtained, the above preprocessing on $S$ takes $O(n\log n)$ time and $O(n)$ space.

This finishes our preprocessing for Lemma~\ref{lem:subsetprocess10}, which uses $O(n^2)$ time and $O(n)$ space. In the following, we describe the query algorithms.

Consider a query line $\ell$. First, for each $\triangle_i$, $1\leq i\leq f$, for each edge $e$ of $\triangle_i$, we determine whether $\ell$ intersects a segment of $S(e)$, which can be done in $O(\log n)$ time by Lemma~\ref{lem:anchor}(1); if the answer is yes, then we halt the entire query algorithm. The total time of this step is $O(f\log n)$; recall that $f\leq 2t$ and $t=\sqrt{n}$. Second, by checking every cell $\triangle_i$, $1\leq i\leq f$, we determine those cells crossed by $\ell$; this takes $O(f)$ time.
For each such cell $\triangle_i$, we determine whether $\ell$ intersects a segment stored in $T_i$. This can be done in the same way as our query algorithm using Chan's partition trees in Section~\ref{sec:randomized}. Starting from the root, we determine the set of cells $\triangle$ of $T_i$ crossed by $\ell$.
For each such cell $\triangle$, if it is a leaf, then we check whether $s$ intersects $\ell$ for each segment $s\in S(\triangle)$. Otherwise, for each edge $e$ of $\triangle$, we use the query algorithm of Lemma~\ref{lem:anchor}(1) to determine whether $\ell$ intersects any segment of $S(e)$. This finishes the algorithm. Lemma~\ref{lem:bound}(2) guarantees that the total query time is $O(\sqrt{n}\log n)$, for there are a total of $O(\sqrt{n})$ cells crossed by $\ell$ and the total number of points of $P$ in those leaf cells crossed by $\ell$ is $O(\sqrt{n}\log n)$ (which implies that the total number of segments stored in those leaf cells crossed by $\ell$ is $O(\sqrt{n}\log n)$).
Therefore, the query time is bounded by $O(\sqrt{n}\log n)$.

\subparagraph{Remark.}
If we set $t$ to  $\sqrt{n}\log n$ as in~\cite{ref:MatousekRa93}, then the query time would become $O(\sqrt{n}\log^2 n)$.
%We explain why we set $t$ to $\sqrt{n\log n}$ instead of $\sqrt{n}\log n$ as in~\cite{ref:MatousekRa93}. The reason is two fold. First, the preprocessing time would go up for another $\sqrt{\log n}$ factor. Second, the first step of the above query algorithm would take $O(\sqrt{n}\log^2 n)$ time, although the second step  would take only $O(\sqrt{n}\log n)$. So setting $t$ to $\sqrt{n\log n}$ minimizes the query time by balancing the time complexities of the two steps.
Note that setting $t=\sqrt{n}\log n$ does not cause any problem for
simplex range searching queries in~\cite{ref:MatousekRa93} because the
issue can be easily resolved by using multilevel data structures. Here
again we cannot effectively use multilevel data structures. On the
other hand, it can easily checked from the proof of
Lemma~\ref{lem:preprowholeP} that smaller $t$ also helps reduce the
preprocessing time. As discussed in Footnote~\ref{foot:10}, $\sqrt{n}$ is
asymptotically the smallest value for $t$ in order to guarantee the
bounds of Lemma~\ref{lem:bound}(2) by following the same proof as in~\cite{ref:MatousekRa93}.
\medskip

Suppose the segments of $S$ are nonintersecting. Consider a query
$\rho$. The algorithm is similar as above but we use the query
algorithm of  Lemma~\ref{lem:anchor}(2) instead on each set $S(e)$. As
a last step, among all segments hit by $\rho$ found by the algorithm as above, we
return the segment whose intersection with $\rho$ is closest to the origin of $\rho$. The
query time is $O(\sqrt{n}\log n)$.

This proves Lemma~\ref{lem:subsetprocess10} and thus Lemma~\ref{lem:subsetprocess}.

%\section{Applications of Subpath Convex Hull Queries}
\section{Concluding Remarks}
\label{sec:app}

We demonstrate several applications of the subpath hull queries where our new result leads to improvement. In each problem, the algorithm needs to preprocess a simple path for subpath hull queries, and the goal of each query is usually to perform certain operations (e.g., one of those listed in Theorem~\ref{theo:subpathhull}) on the convex hull of the query subpath. All algorithms use the previous result of Guibas et al.~\cite{ref:GuibasCo91}. We replace it by our new result in~Theorem~\ref{theo:subpathhull}, which reduces the space of the original algorithm by a $\log\log n$ factor while the runtime is the same as before or even better. In the following, for each problem, we will briefly discuss the previous result and the operations on the convex hull of the query subpath needed in the algorithm; we then present the improvement of using our new result. Refer to the cited papers for the algorithm details of these problems.

\subparagraph{Computing an optimal time-convex hull under the $\boldsymbol{L_p}$ metrics.} Dai et al.~\cite{ref:DaiOp13} presented an algorithm for computing an optimal time-convex hull for a set of $n$ points in the plane under the $L_p$ metrics. The algorithm runs in $O(n\log n)$ time and $O(n\log\log n)$ space. In their algorithm, the operation on the convex hull of the query subpath is the third operation in Theorem~\ref{theo:subpathhull} (called {\em one-sided segment sweeping query} in~\cite{ref:DaiOp13}; see Section 4.2~\cite{ref:DaiOp13}). Using our new result in Theorem~\ref{theo:subpathhull}, the problem can now be solved in $O(n\log n)$ time and $O(n)$ space.

\subparagraph{Computing a guarding set for simple polygons.} Christ et al.~\cite{ref:ChristIm10} studied a new class of art gallery problems motivated by applications in wireless localization. They gave an $O(n\log n)$ time and $O(n\log\log n)$ space algorithm to compute a guarding set for a simple polygon of $n$ vertices (see Corollary 11~\cite{ref:ChristIm10}). In their algorithm, the operation on the convex hull of the query subpath is the third operation in Theorem~\ref{theo:subpathhull}. Using our new result in Theorem~\ref{theo:subpathhull}, the space of the algorithm can be reduced to $O(n)$ while the runtime is still $O(n\log n)$.

\subparagraph{Enclosing rectangles by two rectangles of minimum total area.} Becker et al.~\cite{ref:BeckerAn91} considered the problem of finding two rectangles of minimum total area to enclose a set of $n$ rectangles in the plane. They gave an algorithm of $O(n\log n)$ time and $O(n\log\log n)$ space. In their algorithm, the operation on the convex hull of the query subpath is the third operation in Theorem~\ref{theo:subpathhull}.
Using our new result in Theorem~\ref{theo:subpathhull}, the problem can now be solved in $O(n\log n)$ time and $O(n)$ space.

\subparagraph{Enclosing polygons by two rectangles of minimum total area.} Becker et al.~\cite{ref:BeckerEn96} extended their work above and studied the problem of enclosing a set of simple polygons using two rectangles of minimum total area. They gave an algorithm of $O(n\alpha(n)\log n)$ time and $O(n\log\log n)$ space, where $n$ is the total number of vertices of all polygons and $\alpha(n)$ is the inverse
Ackermann’s function. In their algorithm, the operation on the convex hull of the query subpath is the third operation in Theorem~\ref{theo:subpathhull}.
Using our new result in Theorem~\ref{theo:subpathhull}, the space of the algorithm can be reduced to $O(n)$ while the runtime is still $O(n\alpha(n)\log n)$.

\subparagraph{$\boldsymbol{L_1}$ Top-$\boldsymbol{k}$ weighted sum aggregate nearest and farthest neighbor queries.} Wang and Zhang~\cite{ref:WangOn19} studied top-$k$ aggregate nearest neighbor queries (also called group nearest neighbor queries) using the weighted sum operator under the $L_1$ metric in the plane. They built a data structure of $O(n \log n \log\log n)$ space in $O(n\log n \log\log n)$ time. In their query algorithm, the operation on the convex hull of the query subpath is the third operation in Theorem~\ref{theo:subpathhull} (see Lemma 8~\cite{ref:WangOn19}). Using our new result in Theorem~\ref{theo:subpathhull}, we can reduce both the space and the preprocessing time of their data structure to $O(n\log n)$, while the query time is the same as before.
%Note that Wang and Zhang~\cite{ref:WangOn19} also provided two other data structures with space and query time tradeoff by using the tradeoff of the results of~\cite{ref:GuibasCo91}; both of them can now be replaced by our improved result.
Wang and Zhang~\cite{ref:WangOn19} also considered the farthest neighbor queries and obtained the same result as above using similar techniques, which can also be improved as above by using our new result in Theorem~\ref{theo:subpathhull}.

%%
%% Bibliography
%%

%\newpage
%\bibliographystyle{plain}
\bibliography{reference}

%\newpage
%\appendix
%\section*{APPENDIX}
%
%\medskip
%
%\section{A new section}
%\label{sec:circularhull}

\end{document}